\documentclass{scrartcl}

\usepackage{booktabs} 
\usepackage{amsmath,amssymb,pifont}
\usepackage{amsbsy} 
\usepackage{enumitem}
\usepackage{hyperref}
\usepackage{authblk} 

\usepackage{amsthm}
\newtheorem{theo}{Theorem}
\newtheorem{prop}{Proposition}

\usepackage[xindy]{glossaries}	
\newacronym{FCFS}{FCFS}{First-Come, First-Served}
\newacronym{FIFO}{FIFO}{First-In, First-Out}
\newacronym{PS}{PS}{Processor Sharing}
\newacronym{GPS}{GPS}{Generalized Processor Sharing}
\setkeys{glslink}{hyper=false} 

\usepackage{ulem} 
\usepackage[utf8]{inputenc}

\usepackage{xcolor}
\xdefinecolor{myred}{RGB}{235,62,62}
\xdefinecolor{myblue}{RGB}{92,92,235}
\xdefinecolor{myorange}{RGB}{245,138,42}

\usepackage{graphics,graphicx}
\usepackage{tikz}
\usepackage{subfig}

\usetikzlibrary{calc,positioning,decorations.pathreplacing,shapes}
\tikzset{
	every node/.style = {minimum size = .7cm},
	every edge/.append style = {draw, ->, >=stealth},
	class/.style = {draw, rounded corners=.05cm},
	type/.style = {class, densely dashed},
	server/.style = {draw, circle, minimum size=.9cm},
	pool/.style = {rounded corners=.05cm},
}

\usepackage{pgfplotstable}
\usepackage{pgfplots}
\pgfplotsset{
tikzDefaults/.style={
		enlargelimits=0,
		ymin = 0,
		xmin = 0,
		mark size=2pt,
		line width = 1,
		xticklabel style = {font=\footnotesize, yshift=.1cm},
		yticklabel style = {font=\footnotesize, xshift=.1cm},
		label style = {font=\footnotesize},
		legend style = {font=\footnotesize},
},
tikzLocality/.style={
		enlargelimits=0,
		ymin = 0,
		xmin = 0,
		xticklabel style = {font=\footnotesize, yshift=.1cm},
		yticklabel style = {font=\footnotesize, xshift=.1cm},
		label style = {font=\footnotesize},
		legend style = {font=\footnotesize},
		line width = 1,
		x label style={at={(axis description cs:0.5,-0.1)}},
		y label style={at={(axis description cs:-0.1,.5)}},
	},
	ring/.style = {green!70!black, densely dotted},
	line/.style = {myblue},
	nonloc/.style = {myorange, densely dashed},
	line_c/.style = {myblue, densely dashdotted},
  table/search path={.}
}

\pgfplotsset{compat=newest}

\renewcommand\footnotemark{}
\begin{document}

\title{Performance of Balanced Fairness in Resource Pools: A~Recursive Approach}

\author[1]{Thomas Bonald\thanks{The authors are members of LINCS, see \href{http://www.lincs.fr}{http://www.lincs.fr}.}}
\author[2,1]{C\'eline Comte}
\author[2]{Fabien Mathieu \thanks{Emails: thomas.bonald@telecom-paristech.fr, \{\underline{celine.comte},fabien.mathieu\}@nokia.com}}
\affil[1]{T\'el\'ecom ParisTech, Paris-Saclay University, France}
\affil[2]{Nokia Bell Labs, France}

\subtitle{Sigmetrics 2018 - POMACS (Author version)}
\date{}
\maketitle

\begin{abstract}
  Understanding the performance of a pool of servers is crucial for proper dimensioning.
  One of the main challenges is to take into account the complex interactions between servers that 
  are pooled to process jobs.
  In particular, a job can generally not be processed by any server of the cluster due to various constraints like data locality.
  In this paper, we represent these  constraints  by some assignment graph between jobs and servers.  
  We present a recursive approach   to computing  performance metrics like  mean response times when the server capacities are shared according to balanced fairness.
  While the computational cost of these formulas can be exponential in the number of servers in the worst case,
  we illustrate their practical interest by introducing broad classes of pool structures that can be exactly analyzed in polynomial time.
  This extends considerably the class of models for which explicit performance metrics  are accessible. \\[.2cm]
  \textit{Keywords:} Balanced fairness, parallel computing, performance evaluation. \\[.2cm]
  \textit{CSS 2012:} Mathematics of computing $\to$ Markov processes ; Computer systems organization $\to$ Cloud computing\\[.2cm]
  \textit{DOI:} \url{https://doi.org/10.1145/3154500}
\end{abstract}

\section{Introduction}\label{sec:introduction}
Today, computing infrastructures
consist of thousands of servers interacting in a complex way.
For example, MapReduce is able to process massive data sets by distributing the load 
over a large number of computers where data are located \cite{LY12}. Similarly, the Berkeley Open Infrastructure for Network Computing (BOINC~\cite{BOINC}) offers a generic infrastructure to disseminate various tasks requiring multiple types of resources (CPU, memory, bandwidth, storage,\ldots) over a large pool of heterogeneous devices (computers, game consoles,\ldots).
This approach based on resource pooling also emerges
in content delivery networks, where file replication 
allows requests to be satisfied by multiple servers concurrently. Download tools like JDownloader\footnote{\url{http://jdownloader.org}} are a basic example of this technique: they can accelerate the download of a large file by retrieving different pieces of that file, called chunks, in parallel over multiple hosting servers.
In these new paradigms where resources are not isolated anymore,
the performance of the underlying scheduling policy is still poorly understood,
so that the service providers often rely on over-dimensioning to guarantee proper quality of service.
There is a clear need to better understand the impact of scheduling and load on response times in large
resource pools.

In this paper, we consider a pool of servers whose resources (like CPU or bandwidth) are shared dynamically by ongoing jobs. 
Each job can only be processed by some subset of servers,
which represents various constraints like data locality. 
Grouping jobs in classes so that all jobs of the same class are served by the same subset of servers,
these constraints can be represented as an assignment graph between job classes and servers. 
We assume that the service capacities are shared according to \emph{balanced fairness} \cite{BP03-1} under these constraints, as considered in \cite{SV15,SV16} in the context of content-delivery networks. 

Balanced fairness, which is closely related to proportional fairness \cite{M07},  has the double practical interest of leading to explicit expressions for the stationary distribution, due to the reversibility of the underlying Markov process, and to have the  insensitivity property, in the sense that the stationary distribution does not depend on the service-time distribution beyond the mean. 
Thus, it is often considered as a desirable sharing objective, yielding simple and robust performance results.
Moreover, it has recently been shown  that balanced fairness naturally emerges from some simple scheduling policies. 
A first example is the \emph{redundant requests} approach introduced in 
 \cite{G16,G17-1,G17-2}, where a given job is replicated over all servers that can process it, and the first instance to complete stops the others. Redundant requests are well suited to jobs that cannot be parallelized due to their nature or time scale, like the elementary tasks of a fine-grain computation running in a computer cluster. When jobs can be sliced into multiple chunks, balanced fairness can also be achieved under \emph{parallel processing}, as shown in  \cite{BC17-1}. This is typically the case in BOINC, where a given task can be split into work units, or in content distribution networks, where small parts of a file can be retrieved independently.
 
Unfortunately, just knowing the expression of the stationary distribution is not enough to derive performance metrics like the mean response times, even just numerically.
As usual, this requires the computation of the normalizing constant, which is a notoriously hard problem \cite{H85}.
This is why existing results consider either small systems (e.g., 3 servers)
or symmetric systems (e.g., all servers have the same service rate and each job can be processed by $k$ servers chosen uniformly at random) \cite{G16,G17-1,G17-2}.
The notion of poly-symmetry has recently been introduced to enlarge the class of tractable models but it still relies on some specific (poly-)symmetric properties of the system \cite{BCV17}.

The main contribution of this paper is a new recursive approach for computing
the normalizing constant (equivalently, the probability that the system is empty)
and thus the mean response times of systems under balanced fairness.
Our recursive formula applies to any constraint structure, that is, any assignment graph between job classes and servers.
Of course, its complexity depends on the degree of symmetry of the system, but it is not limited to strictly symmetric or poly-symmetric models.
In particular, we exhibit two large classes of assignment graphs
where the complexity of the formula is polynomial in the number of servers, instead of exponential:
\emph{randomized} assignments and \emph{local} assignments.
We show that these classes can be seen as generalizations of examples previously identified and analyzed in \cite{G17-1,G17-2} in the context of redundant requests,
and we illustrate them by a number of new examples that are practically interesting and computationally tractable.
Thus, our work extends considerably the set of systems for which closed-form performance metrics are accessible,
which will hopefully provide very useful insights into the behavior of  large-scale resource pools.

The rest of the paper is organized as follows.
In the following section,
we introduce the model of a resource pool under balanced fairness
and describe the sequential implementation of this policy.
The recursive formula is presented in Section \ref{sec:results-statement}.
The applications to {randomized} assignment and {local} assignment
are presented in Sections \ref{sec:randomizeds} and \ref{sec:ranged-clusters}, respectively.
Numerical results are provided in Section \ref{sec:numerical-evaluation}.
Section \ref{sec:conclusion} concludes the paper.

\section{Resource pool}\label{sec:model}

We consider a model of resource pool that applies to a large variety of systems, like computer clusters or content delivery networks.

\subsection{Model}\label{sec:cluster-model}

Consider a resource pool with $I$ job classes and $K$ servers.
The sets of class and server indices are denoted by $\mathcal{I}$ and $\mathcal{K}$, respectively.
For each $i \in \mathcal{I}$, class-$i$ jobs enter the system as a Poisson process with rate $\lambda_i$.
The corresponding vector of arrival rates is denoted by ${\lambda} = (\lambda_i: i \in \mathcal{I})$.
Each job leaves the system as soon as its service is complete.
For each $k \in \mathcal{K}$, the service capacity of server $k$ is denoted by $\mu_k$.

The class of a job defines the set of servers that are assigned to this job. 
It may be determined by practical constraints like data locality
or result from some  load balancing scheme, as explained in Sections \ref{sec:randomizeds} and \ref{sec:ranged-clusters}.
For each $i \in \mathcal{I}$,
$\mathcal{K}_i \subset \mathcal{K}$
denotes the set of servers assigned
to each job of class $i$.
Reciprocally, for each
$k \in \mathcal{K}$,
$\mathcal{I}_k \subset \mathcal{I}$
denotes the set of job classes that are assigned server $k$,
i.e., $i \in \mathcal{I}_k$ if and only if $k \in \mathcal{K}_i$.
All assignments can be described by a bipartite
graph between classes and servers, where there is an edge between class $i$ and server $k$
if and only if server $k$ is assigned to class $i$.
We assume without loss of generality that each server is assigned at least one class.
An example,  referred to as the M model  \cite{G16}, is shown in Figure \ref{fig:toy}: there are  3 servers and 2 job classes;
servers $1$ and $2$ are dedicated to job classes $1$ and $2$ respectively,
while server $3$ can serve both classes.

\begin{figure}[ht]
\centering
\begin{tikzpicture}[draw, node distance=2cm,auto,>=stealth]
  \pgfmathtruncatemacro\n{2}

  \pgfmathtruncatemacro\s{\n+1}
  \foreach \i in {1,...,\s}{
    \pgfmathsetmacro\x{(\i*1.5)-1.5}
    \node[server] (s\i) at (\x,0) {};
  }
  \foreach \i in {1,...,\n}{
    \pgfmathsetmacro\x{(\i*1.5)-.75}
    \node[class] (l\i) at (\x,1.6) {$\lambda_\i$};
    \pgfmathtruncatemacro\j{\i+1}
    \draw (l\i) edge (s\i) edge (s\j);
  }

  \node at (s1) {$\mu_1$};
  \node at (s2) {$\mu_3$};
  \node at (s3) {$\mu_2$};

  \node[anchor=east] at ($(s3)+(2.2,0)$) {Servers};
  \node[anchor=east] at ($(s3)+(2.2,1.6)$) {Job classes};
\end{tikzpicture}
	\caption{Example of assignment graph: the M model}
		\label{fig:toy}
\end{figure}
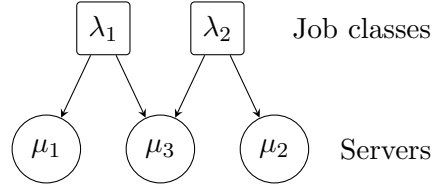

The system state is described by the vector $x = (x_i : i \in \mathcal{I})$ of numbers of jobs of each class.
Throughout the paper,  a server is  said to be idle if it has no job to process while a class is  said to be idle if it has no ongoing job (neither queued nor in service).

Our main notations are summarized in Table \ref{tab:notation}.

\begin{table}[ht]
  \centering
  \begin{tabular}{|>{\centering}m{.08\linewidth}|p{.33\linewidth}|>{\centering}m{.08\linewidth}|p{.33\linewidth}|}
    \hline 
    \multicolumn{2}{|c|}{Servers}
    & \multicolumn{2}{c|}{Job classes} \\ 
    \hline 
    $\mathcal{K}$ & Set of servers &
    $\mathcal{I}$ & Set of classes \\ 
    \hline 
    $K$ & Number of servers &
    $I$ & Number of classes \\ 
    \hline 
    $\mu_k$ & Service rate of server $k$ &
    $\lambda_i$ & Arrival rate of class-$i$ jobs \\ 
    \hline 
    $M(\mathcal{L})$ & Total service rate of servers $\mathcal{L}\subset \mathcal{K}$ &
    $\Lambda(\mathcal{A})$ &  Total arrival rate of classes $\mathcal{A} \subset \mathcal{I}$ \\
    \hline 
    \multicolumn{4}{c}{ } \\
    \hline
    \multicolumn{4}{|c|}{Assignment graph} \\ 
    \hline 
    $\mathcal{K}_i$ & \multicolumn{3}{l|}{Servers assigned to class-$i$ jobs} \\ 
    \hline 
    $\mathcal{I}_k$ & \multicolumn{3}{l|}{Classes that are assigned server $k$} \\ 
    \hline
    $ |-k $ & \multicolumn{3}{p{.82\linewidth}|}{System reduced to servers in $ \mathcal{K}\setminus \{k\} $ and to jobs that are not assigned server $ k $} \\
    \hline
    $ |\mathcal{L} $ & \multicolumn{3}{p{.82\linewidth}|}{System reduced to servers in $ \mathcal{L}$ and to jobs that are assigned servers in $ \mathcal{L} \subset \mathcal{K} $ only} \\
    \hline
    \multicolumn{4}{c}{ } \\
    \hline
    \multicolumn{4}{|c|}{Performance metrics} \\ 
    \hline 
    $\psi$ & \multicolumn{3}{l|}{Probability that the system is idle} \\
    \hline
    $ L $ & \multicolumn{3}{l|}{Mean number of jobs} \\
    \hline
  \end{tabular} 
  \caption{\label{tab:notation} Table of notation.}
\end{table}

\subsection{Balanced fairness}
\label{sec:bf}

We assume that 
the resources are allocated according to balanced fairness  \cite{BP03-1}.
Each server shares its service capacity among its assigned jobs, in a way that depends on the system state $x$.
For each $i \in \mathcal{I}$,
we denote by $\phi_i(x)$  the total service rate received  by class-$i$ jobs  in state $x$. This is the sum of all service rates allocated by servers in $\mathcal{K}_i$ to class-$i$ jobs in state $x$. 
Under balanced fairness, all class-$i$ jobs are assumed to be served at the same rate, namely  $\phi_i(x) / x_i$ in any state $x$ such that  $x_i > 0$ (we detail in \S \ref{subsec:sequential} how this allocation may be achieved in practice.).
We adopt the convention that $\phi_i(x) = 0$ if $x_i = 0$.

The capacity set $\mathcal{C}$ of the system is defined as the set of all feasible allocations,
$$
\mathcal{C} = \left\{
  \phi \in \mathbb{R}_+^\mathcal{I}:~
  \sum_{i \in \mathcal{A}} \phi_i
  \le M\left( \bigcup_{i \in \mathcal{A}} \mathcal{K}_i \right),~
  \forall \mathcal{A} \subset \mathcal{I}
\right\},
$$
where for any $ \mathcal{L} \subset  \mathcal{K}$,
$M( \mathcal{L} ) = \sum_{k\in  \mathcal{L} } \mu_k$ denotes the total service capacity of the servers in $ \mathcal{L} $.
Under balanced fairness, the service rates are given by
\begin{equation*}
  \phi_i(x) = \frac{ \Phi(x-e_i) }{ \Phi(x) },
  \quad \forall x \in \mathbb{N}^\mathcal{I},
  \quad \forall i: x_i > 0,
\end{equation*}
where $e_i$ is the $I$-dimensional vector with $1$ in component $i$ and $0$ elsewhere,
and $\Phi$ is the balance function,
defined recursively by $\Phi(0) = 1$ and
\begin{equation}
  \label{eq:Phi}
  \Phi(x) = \frac1{ M\left( \bigcup_{i:x_i > 0} \mathcal{K}_i \right) }
  \sum_{i: x_i > 0} \Phi(x-e_i),
  \quad \forall x \neq 0.
\end{equation}
Observe that the corresponding vector $\phi(x)$ belongs to the capacity set $\mathcal{C}$. Moreover,
$$
\sum_{i \in \mathcal{I}} \phi_i(x) = M\left( \bigcup_{i:x_i > 0} \mathcal{K}_i \right),
$$
so that each non-idle server  is fully utilized.

Now if job sizes are i.i.d.~exponential with unit mean, the underlying Markov process is reversible and the stationary measure of the system state is  given by
\begin{equation}
  \label{eq:pix}
  \pi(x) = \pi(0) \Phi(x) \lambda^x,
  \quad \forall x \in \mathbb{N}^\mathcal{I},
\end{equation}
where we use the notation $\lambda^x = \prod_{i \in \mathcal{I}} {\lambda_i}^{x_i}$.

\subsection{Sequential implementation of balanced fairness}\label{subsec:sequential}

Balanced fairness assumes that each server has the ability to arbitrarily split its capacity. Yet, many real-life servers can only process jobs sequentially.
We now show how to conciliate the two viewpoints by considering a sequential implementation that behaves like balanced fairness,
although each server processes only one job at a time, in \gls{FCFS} order.
In details, we exhibit two variants, introduced respectively in \cite{G16} and \cite{BC17-1}:

\begin{description}
  \item[Redundant requests.]
   Each class-$i$ job is replicated over all servers in $\mathcal{K}_i$. When a job is in service on several servers at the same time, each of these servers  works on a copy of the job, independently of the other servers.
    The service times of  the  copies of the same job are independent and exponential, with parameter $\mu_k$ at server $k$.
    A job leaves the system whenever any of its copies has completed its service, so after an exponential time with parameter 
    $\sum_{k\in \mathcal{A}}\mu_k$, where $\mathcal{A}$ is the (time-varying) set of servers working on the same job;
    the services of all other copies of this job are then interrupted
    and these copies are removed from the corresponding queues.
  \item[Parallel processing.]
    When a job is in service on several servers at the same time,
    these servers are pooled to process (a single copy of) this job in parallel.
    The parallel processing is assumed to be perfect,
    so that the service rate of a job is the sum of the service capacities of the servers that are processing it.
    The  job sizes are i.i.d.~exponential  with unit mean, so that the service time of a job  is exponential with parameter $\sum_{k\in \mathcal{A}}\mu_k$, where $\mathcal{A}$ is the (time-varying) set of servers processing this job.
\end{description}
These two variants are described by the same Markov process but rely on different assumptions. In the first model, the servers are independent and the work done by some server on some copy of a job cannot be used by the other servers; the gain of redundancy relies on the independence of the service times of the copies of each job and their specific (exponential) distribution. In the second model, the servers need to coordinate to process the same job, so that the work done by a server  doesn't need to be done by other servers. This coordination can be achieved by dividing each job into small chunks that are distributed dynamically among active servers, say by some master server elected at the job arrival.

The system state defines a Markov process provided it includes  the arrival order of jobs.
We consider the sequence $c = (c_1,\ldots,c_n)$ of job classes in order of arrival of the jobs,
where $n$ is the total number of jobs in the system
and $c_p \in \mathcal{I}$ is the class of the job in position $p$, for each $p = 1,\ldots,n$, so that job in position 1 is the oldest job of the system.
The state space is the set $\mathcal{I}^*$ of all finite sequences on $\mathcal{I}$.
The corresponding queueing model is an order-independent queue, as introduced in \cite{B96,K11}.
The  system state has the following stationary measure $ \pi $ \cite{B96,K11,G16,BC17-1}:
\begin{equation}
  \label{eq:pic}
  \pi(c) = \pi(\emptyset) \prod_{p=1}^n \frac{\lambda_{c_p}}{ M\left( \bigcup_{q=1}^p \mathcal{K}_{c_{q}} \right) },
  \quad \forall c = (c_1,\ldots,c_n) \in \mathcal{I}^*.
\end{equation}
Now consider the aggregate state $x = (x_i: i \in \mathcal{I})$ of the number of jobs of each class,
independently of the order of arrival of these jobs.
With a slight abuse of notation, we also denote by $\pi$ the stationary measure of this aggregate state and get
\begin{equation}
  \label{eq:piagg}
  \pi(x) = \sum_{c:|c|=x} \pi(c),
  \quad \forall x \in \mathbb{N}^\mathcal{I},
\end{equation}
where $|c|$ denotes the aggregate state associated to state $c$.
The following result shows that this stationary measure is also that obtained under balanced fairness.
The proof is borrowed from \cite{B96,K11,BC17-1}.

\begin{prop} 
 \label{equivalence}
  The stationary measures \eqref{eq:pix} and \eqref{eq:piagg} coincide.
\end{prop}

\begin{proof}
  For each $x \in \mathbb{N}^N$, the stationary measure \eqref{eq:pic} satisfies:
  \begin{align*}
    \sum_{c:|c|=x} \pi(c)
    &= \sum_{c:|c|=x} \pi(\emptyset) \prod_{p=1}^n \frac{\lambda_{c_p}}{ M\left( \bigcup_{q=1}^p \mathcal{K}_{c_{q}} \right) }
    = \pi(0)
    \left( \sum_{c:|c|=x} \prod_{p=1}^n \frac1{ M\left( \bigcup_{q=1}^p \mathcal{K}_{c_{q}} \right)} \right)
    \lambda^x.
  \end{align*}
  The result follows by letting
  \begin{align*}
    \Phi(x)
    &= \sum_{c:|c|=x} \prod_{p=1}^n \frac1{M \left( \bigcup_{q=1}^p \mathcal{K}_{c_{q}} \right)}
  \end{align*}
  and partitioning the sum depending on the value of $c_n$, which gives \eqref{eq:Phi}.
\end{proof}

In view of Proposition \ref{equivalence}, the results derived in the rest of the paper
equally predict the performance under balanced fairness and the sequential scheduling described above.
Proposition \ref{prop:average} goes one step further
by showing that balanced fairness is the average per-class resource allocation obtained under the sequential scheduling.
The proof can be found in \cite{BC17-1}.

\begin{prop}
  \label{prop:average}
  For each $i \in \mathcal{I}$,
  the mean service rate of class-$i$ jobs under the above sequential scheduling, conditioned on the number of jobs of each class in the system,
  is the service rate obtained under balanced fairness:
  $$
  \phi_i(x) = \sum_{c:|c|=x} \frac{\pi(c)}{\pi(x)}
  \sum_{\substack{p=1 \\ c_p = i}}^n
  \left(
  \mu(c_1,\ldots,c_p) - \mu(c_1,\ldots,c_{p-1})
  \right).
  $$
\end{prop}

This proposition relates the average per-class service rates
but it does not say anything about the rate perceived by each job.
However, it is observed in \cite{BC17-1} that
balanced fairness can be effectively realized
in sequential systems by enforcing frequent job interruptions and resumptions
on top of the \gls{FCFS} scheduling.
This extends the way \gls{PS} policy can be implemented by a round-robin scheduler in the single-server case.
In the queueing model,
these interruptions and resumptions are modeled by adding random routing,
which leaves unchanged the stationary measure of the system state.
If the interruptions are enough frequent,
all jobs tend to be served at the same rate on average,
which is precisely the service rate $\phi_i(x) / x_i$ considered in balanced fairness.

Additionally, interrupting jobs frequently allows to reach some approximate insensitivity to the job size distribution.
In the limit,  the resource allocation is exactly balanced fairness and
the assumption of exponential service times is not required anymore.
In fact, it is not even necessary to assume unit mean job sizes.
The results remain the same for any mean job sizes provided $\lambda_i$ is interpreted as the traffic intensity of class $i$
(quantity of work brought by class-$i$ jobs per time unit)
rather than the arrival rate of class-$i$ jobs.

\subsection{Stability condition}\label{sec:stab}

It is known that balanced fairness stabilizes the system whenever
the vector of arrival rates $\lambda$ lies in the interior of the capacity set ${\mathcal C}$ \cite{BMPV06}.
This shows that the system is stable whenever
\begin{equation*}
  \Lambda\left( \mathcal{A} \right) < M\left( \bigcup_{i \in \mathcal{A}} \mathcal{K}_i \right),
  \quad \forall \mathcal{A} \subseteq \mathcal{I} ~\text{with}~ \mathcal{A} \neq \emptyset,
\end{equation*}
where for any $ \mathcal{A} \subset  \mathcal{I}$, $\Lambda(\mathcal{A})   = \sum_{i\in  \mathcal{A} } \lambda_i$ denotes the total arrival rate of classes $ \mathcal{A} $.
Equivalently, focusing on servers instead of jobs,  the stability condition can be written 
\begin{equation}
  \label{eq:stability}
  \Lambda\left( \mathcal{I} \setminus \bigcup_{k \in \mathcal{K} \setminus \mathcal{L}} \mathcal{I}_k \right)
  < M(\mathcal{L}),
  \quad \forall \mathcal{L} \subseteq \mathcal{K} ~\text{with}~ \mathcal{L} \neq \emptyset.
\end{equation}
We assume that this condition is satisfied in the following
and we let $\pi$ denote the stationary distribution of the system state.

\subsection{Performance metrics}\label{sec:steady-state-behavior}

We are interested in the mean response time $T_i$ of each job of class $i$.
By Little's law, we have $T_i = L_i/\lambda_i$, where $L_i$ denotes the mean number  of jobs of class $i$.
It follows from \eqref{eq:pix} that
$$
L_i = \sum_{x \in \mathbb{N}^\mathcal{I}} x_i \pi(x)
= \pi(0) \sum_{x \in \mathbb{N}^\mathcal{I}} x_i \Phi(x) \lambda^x.
$$
This expression involves the probability
\begin{equation*}
  \psi
  = \pi(0)
  = \frac 1 {\sum_{x \in \mathbb{N}^\mathcal{I}} \Phi(x) {\lambda}^x}
\end{equation*}
that the system is empty,
which is the inverse of the normalizing constant.
Considering $\psi$ as a function of $\lambda$, the mean numbers of jobs follow by taking the derivative.
This result was already stated for balanced fairness in \cite{BV04}.
The proof is recalled for the sake of completeness.

\begin{prop}
  For each $i \in \mathcal{I}$, we have
    \begin{equation}
      \label{eq:defmean}
      L_i = \psi \lambda_i  \frac{\partial \left(\frac 1 \psi\right)}{\partial \lambda_i},
      \quad \forall i \in \mathcal{I}.
    \end{equation}
\end{prop}

\begin{proof}
  Let $i \in \mathcal{I}$.
  We have successively
  \begin{align*}
    L_i
    &= \psi \sum_{x \in \mathbb{N}^\mathcal{I}} x_i \Phi(x) {\lambda}^x
    = \psi\lambda_i  \sum_{x: x_i > 0} x_i \Phi(x) {\lambda}^{x-e_i}
    =\psi \lambda_i  \frac{\partial \left(\frac 1 \psi\right)}{\partial \lambda_i}.
  \end{align*}
\end{proof}

\section{Recursive formula}\label{sec:results-statement}

We now present the main result of the paper, that is a recursive formula for computing the probability $\psi$ that the system is empty. 
We then derive other recursive formulas for the mean number of jobs of each class.

\subsection{Conditioning}
\label{subsec:conditioning}

The key idea of the recursion is to  condition on the fact that some server $k \in \mathcal{K}$ is idle. Observing that server $k$ is idle if and only if there are no active  jobs of classes 
$\mathcal{I}_k$, this occurs with probability:
\begin{equation}\label{eq:psik}
\psi_k = \psi  \sum_{x:\sum_{i \in \mathcal{I}_k} x_i = 0} \Phi(x) {\lambda}^x.
\end{equation}

Now consider the same pool of servers but without any traffic generated by jobs of classes $\mathcal{I}_k$, that is, for the vector of  arrival rates   ${\lambda}_{|-k}$  defined by  $\lambda_{i|-k} = \lambda_i 1_{\{i\not \in  \mathcal{I}_k\}}$ for all $i\in \mathcal{I}$.
The stationary distribution of the state of this reduced system is
$$
\pi_{|-k}(x) = \psi_{|-k} \Phi(x) {\lambda_{|-k}}^x,
$$
where $\psi_{|-k}$ is the probability that this system is empty, given by
$$
\psi_{|-k} = \frac 1{ \sum_{x} \Phi(x) {\lambda^x_{|-k}}}.
$$
Note that $\psi_{|-k} = 1$ if $\mathcal{I} = \mathcal{I}_k$.
In view of \eqref{eq:psik}, we have:
\begin{equation}
  \label{eq:subsystem}
  \psi = \psi_k \psi_{|-k},
\end{equation}
so that $\psi_{|-k}$ can also be interpreted as the probability that the initial system is empty given that server $k$ is idle.
Similarly, $\pi_{|-k}$ can be viewed as the conditional stationary distribution of the system state, given that server $k$ is idle.
All our results rely on this simple but powerful observation.

\subsection{Probability of an empty system}
\label{subsec:normalization}

The following theorem relates the probability that the system is empty
to the conditional probability that it is empty, given that some server is idle.
This gives a method to compute $\psi$ recursively.

\begin{theo}
  \label{theo:empty}
  The probability that the system is empty is given by
    \begin{equation}
      \label{eq:emptyderivative}
      \psi = \frac
      { M(\mathcal{K}) - \Lambda(\mathcal{I}) }
      { \sum_{k \in \mathcal{K}} \frac{\mu_k}{\psi_{|-k}} }\text{,}
    \end{equation}
  which can also be expressed as
  \begin{equation}
    \label{eq:empty}
    \psi = (1 - \rho) \frac{ M(\mathcal{K}) }{ \sum_{k \in \mathcal{K}} \frac{\mu_k}{\psi_{|-k}} },
  \end{equation}
  where $\rho = \frac{ \Lambda(\mathcal{I}) }{ M(\mathcal{K}) }$ is the overall load of the system.
\end{theo}

\begin{proof}
	We first write the conservation equation, which states that the total arrival rate  must be equal to the total average service rate  (accounting for idle periods):
	\begin{equation*}
    \Lambda(\mathcal{I}) = \sum_{k \in \mathcal{K}} \mu_k (1-\psi_k)\text{,}
	\end{equation*}
  that is,
	\begin{equation}
    \label{eq:conservation}
    \sum_{k \in \mathcal{K}} \mu_k \psi_k
    = M(\mathcal{K}) - \Lambda(\mathcal{I}).
  \end{equation}
  Combining \eqref{eq:subsystem} and \eqref{eq:conservation} yields \eqref{eq:emptyderivative}, from which \eqref{eq:empty} follows.
\end{proof}

The probability $\psi$ can be computed by recursively applying \eqref{eq:emptyderivative} or \eqref{eq:empty},
conditioning on the server activity.
The base case of the recursion corresponds to a pool without any input, which is idle with probability $1$.
For each set $\mathcal{L} \subset \mathcal{K}$ of active servers,
we need to evaluate $M(\mathcal{L})$ and
$$
\Lambda\left( \mathcal{I} \setminus \bigcup_{k \in \mathcal{K} \setminus \mathcal{L}} \mathcal{I}_k \right),
$$
which takes $O(I + K)$ operations, where $I$ is the number of job classes and $K$ the number of servers.
The overall complexity is thus in $O\left( (I+K) 2^K \right)$ in the worst case.
Sections \ref{sec:randomizeds} and \ref{sec:ranged-clusters} give practically interesting examples 
where the complexity is polynomial in the number of servers thanks to symmetries or topological properties.

Theorem \ref{theo:empty} and its proof reveal some important properties of the system, which we briefly detail here.

\paragraph{Stability}
The stability condition \eqref{eq:stability} appears when expanding recursion \eqref{eq:emptyderivative}:
the system is stable if and only if
its conditional probability of being empty is positive,
given any set of idle servers.

\paragraph{Resource pooling}
Assume complete resource pooling, that is, $\mathcal{K}_i = \mathcal{K}$ for all $i \in \mathcal{I}$. Then, balanced fairness coincides with \gls{PS} policy while the sequential scheduling coincides with \gls{FIFO} policy, where whenever a job is served, it is served by all servers. In other words, the system boils down to an M/M/1 queue of load $\rho$.

This queue is empty with probability $1 - \rho$, which is the first factor in \eqref{eq:empty}.
In general, the second factor, which can be written:
$$
\frac{\sum_{k \in \mathcal{K}} \mu_k}{ \sum_{k \in \mathcal{K}} \frac{\mu_k}{\psi_{|-k}} },
$$
quantifies the overhead due to  incomplete resource pooling. This is the harmonic mean of the conditional probabilities $\psi_{|-k}$ for $k \in \mathcal{K}$,
weighted by the service rates $\mu_k$ for $k \in \mathcal{K}$.

\paragraph{Activity rates}
For each $i \in \mathcal{I}$,
denote by $\psi_i$ the probability that class $i$ is idle
and by $\psi_{|-i}$ the probability that the system without class $i$ is empty.
As in \S \ref{subsec:conditioning},
one can show that $\psi =  \psi_i \psi_{|-i}$.
Applying  \eqref{eq:emptyderivative} to both $\psi$ and $\psi_{|-i}$, we get
\begin{equation}
\label{eq:psi_vs_rho}
	\begin{split}
\psi_i
&= \frac\psi{\psi_{|-i}}
= \frac{M(\mathcal{K}) - \Lambda(\mathcal{I})}{M(\mathcal{K}) - \Lambda(\mathcal{I} \setminus \{i\})}
 \frac
{ \sum_{k \in \mathcal{K}} \frac{\mu_k}{\psi_{|-k,i}} }
{ \sum_{k \in \mathcal{K}} \frac{\mu_k}{\psi_{|-k}} }
= \left( 1 - \rho_i \right)
\frac
{ \sum_{k \in \mathcal{K}} \frac{\mu_k}{\psi_{|-k,i}} }
{ \sum_{k \in \mathcal{K}} \frac{\mu_k}{\psi_{|-k}} },
	\end{split}
\end{equation}
where $\rho_i = \frac{\lambda_i}{M(\mathcal{K}) - \Lambda(\mathcal{I} \setminus \{i\})}$
corresponds the load associated to class $i$
and $\psi_{|-k,i}$ denotes the conditional probability that the system is empty given that class $i$ and server $k$ are idle,
for each $k \in \mathcal{K}$.
Again, the first factor, $1-\rho_i$,  is the probability that class $i$ is idle under complete resource pooling while the second factor quantifies the overhead due to incomplete resource pooling.

\paragraph{Server occupancies}
In view of  \eqref{eq:subsystem}, the recursion  \eqref{eq:empty} applied to both $\psi$ and $\psi_{|-k}$ gives an effective way of computing $\psi_k$,
the probability that server $k$ is idle, for each $k \in \mathcal{K}$.
From this we can easily compute the mean number of active servers in steady state,
given by $K - \sum_{k \in \mathcal{K}} \psi_k$.

\subsection{Mean number of jobs}
\label{subsec:perf}

We now extend the recursion of Theorem \ref{theo:empty} to get the mean number of jobs of each class in the system, from which we can derive the mean response times.
The notations are the same as above.

\begin{theo}
  \label{theo:mean}
  For each $i \in \mathcal{I}$, the mean number of class-$i$ jobs in the system is given by
  \begin{equation}
    \label{eq:meani_s}
    L_i
    =  \frac{\lambda_i + \sum_{k \in \mathcal{K} \setminus \mathcal{K}_i} \mu_k \psi_k L_{i|-k}} 
    { M(\mathcal{K}) - \Lambda(\mathcal{I}) }\text{,}
  \end{equation}
  and the mean number of jobs in the system is
  \begin{equation}
    \label{eq:mean_s}
    L
    =  \frac{\Lambda(\mathcal{I}) + \sum_{k \in \mathcal{K}} \mu_k \psi_k L_{|-k}} 
    { M(\mathcal{K}) - \Lambda(\mathcal{I}) }\text{.}
  \end{equation}

  Equivalent expressions are
  \begin{equation}
    \label{eq:meani}
    L_i
    = \frac{\rho_i}{1-\rho_i}
    + \frac1{1-\rho}
    \frac{ \sum_{k \in \mathcal{K} \setminus \mathcal{K}_i} \mu_k \psi_k L_{i|-k} }{ M(\mathcal{K}) }
  \end{equation}
  and
  \begin{equation}
    \label{eq:mean}
    L = \frac\rho{1-\rho}
    + \frac1{1-\rho}
    \frac{ \sum_{k \in \mathcal{K}} \mu_k \psi_k L_{|-k} }{ M(\mathcal{K}) },
  \end{equation}
  \noindent where
  $\rho_i = \frac{\lambda_i}{ M(\mathcal{K}) - \Lambda(\mathcal{I} \setminus \{i\}) }$
  is the load associated to class $i$ and $\rho = \frac{ \Lambda(\mathcal{I}) }{ M(\mathcal{K}) }$ is the overall load in the system. 
\end{theo}

\begin{proof}
  Let $i \in \mathcal{I}$. In view of \eqref{eq:emptyderivative}, we have
  \begin{align*}
    \frac{\partial}{\partial \lambda_i} \left( \frac1{\psi} \right)
    ={} &\frac1{ \left( M(\mathcal{K}) - \Lambda(\mathcal{I}) \right)^2 }
    \sum_{k \in \mathcal{K}} \frac{\mu_k}{\psi_{|-k}}
    + \frac1{ M(\mathcal{K}) - \Lambda(\mathcal{I}) }
    \sum_{k \in \mathcal{K} \setminus \mathcal{K}_i} \mu_k
    \frac\partial{\partial \lambda_i} \left( \frac1{\psi_{|-k}} \right).
  \end{align*}
  We recognize the expression of the inverse of $\psi$ given by \eqref{eq:emptyderivative} in the first term.
  Injecting this in \eqref{eq:defmean} yields
  \begin{align*}
    L_i
    = \frac{\lambda_i}{ M(\mathcal{K}) - \Lambda(\mathcal{I}) }
    + \frac1{ M(\mathcal{K}) - \Lambda(\mathcal{I}) }
    \sum_{k \in \mathcal{K} \setminus \mathcal{K}_i} \mu_k \lambda_i \psi
    \frac\partial{\partial \lambda_i} \left( \frac1{\psi_{|-k}} \right).
  \end{align*}
  Additionally, for each $k \in \mathcal{K} \setminus \mathcal{K}_i$, we have by \eqref{eq:defmean} and \eqref{eq:subsystem},
  \begin{align*}
    \lambda_i \psi \frac\partial{\partial \lambda_i} \left( \frac1{\psi_{|-k}} \right)
    &= \frac\psi{\psi_{|-k}} \times \lambda_i \psi_{|-k}
    \frac\partial{\partial \lambda_i} \left( \frac1{\psi_{|-k}} \right)
    = \psi_k L_{i|-k}\text{.}
  \end{align*}
  Hence we obtain \eqref{eq:meani_s}. 
  \eqref{eq:mean_s} follows by observing that
  \begin{align*}
    \sum_{i\in {\mathcal I}} \sum_{k \in \mathcal{K} \setminus \mathcal{K}_i} \mu_k \psi_k L_{i|-k}
      =
    \sum_{i\in {\mathcal I}} \sum_{k \in \mathcal{K}} \mu_k \psi_k L_{i|-k}
      =
    \sum_{k \in \mathcal{K}} \mu_k \psi_k  \sum_{i\in {\mathcal I}}  L_{i|-k}
     = 
    \sum_{k \in \mathcal{K}} \mu_k \psi_k   L_{|-k}.
  \end{align*}

  Finally,  \eqref{eq:meani} is a simple rewriting of  \eqref{eq:meani_s}, using
  $$
  \frac{\lambda_i}{ M(\mathcal{K}) - \Lambda(\mathcal{I})} = \frac{\lambda_i}{ M(\mathcal{K}) - \Lambda(\mathcal{I} \setminus \{i\}) - \lambda_i }
  = \frac{\rho_i}{1-\rho_i}
  $$
  \noindent and
  \begin{align*}
    \frac{1}{ M(\mathcal{K}) - \Lambda(\mathcal{I})} = 
    \frac{M(\mathcal{K})}{ M(\mathcal{K}) - \Lambda(\mathcal{I})} . \frac{1}{ M(\mathcal{K}) } =
    \frac{1}{1-\rho}.\frac{1}{ M(\mathcal{K}) }
    \text{.}
  \end{align*}
  \eqref{eq:mean} follows by summation, observing that
  $$
  \sum_{i \in \mathcal{I}} \frac{\rho_i}{1-\rho_i}
  =  \sum_{i \in \mathcal{I}} \frac{\lambda_i}{M(\mathcal{K}) - \Lambda(\mathcal{I})}
  = \frac\rho{1-\rho}.
  $$
\end{proof}

As in Theorem \ref{theo:empty}, the complexity of each of these recursive formulas is $O\left( (I+K) 2^K \right)$ in the worst case but, again, polynomial in a number of practically interesting cases. 
Moreover, expressions \eqref{eq:meani} and  \eqref{eq:mean} reveal the impact of incomplete resource pooling on performance, as the first terms of each expression, $\rho_i / (1-\rho_i)$
and  $\rho / (1-\rho)$, are the number of class-$i$ jobs and the total number of jobs under complete resource pooling.

\subsection{Toy example}\label{sec:toy-example}

Before presenting practically interesting applications of our recursive formula  in the following two sections, we illustrate it on the M model pictured in Figure \ref{fig:toy}.
Its analysis, already performed in \cite{G16}, is now simplified by a direct application of Theorems \ref{theo:empty} and \ref{theo:mean}.
Let $\lambda = \lambda_1+\lambda_2$ and $\mu=\mu_1+\mu_2+\mu_3$ be the total arrival and service rates in the system. 
The system load is $\rho = \lambda / \mu$.
Using the fact that
$$
\psi_{|-1} = 1 - \frac{\lambda_2}{\mu_2+\mu_3},\quad \psi_{|-2} = 1 - \frac{\lambda_1}{\mu_1+\mu_3},\quad \psi_{|-3} = 1, 
$$
we get from \eqref{eq:empty}:
$$
\psi = (1-\rho) \frac{\mu}{\mu'},
$$
where
$$
\mu' =
\mu_1 \frac{\mu_2 + \mu_3}{\mu_2 + \mu_3 - \lambda_2}
+ \mu_2 \frac{\mu_1 + \mu_3}{\mu_1 + \mu_3 - \lambda_1}
+ \mu_3.
$$

For the mean number of jobs in the system, we have 
$$
L_{|-1} = \frac{\lambda_2}{\mu_2+\mu_3-\lambda_2},\quad L_{|-2} = \frac{\lambda_1}{\mu_1+\mu_3-\lambda_1},\quad L_{|-3} = 0, 
$$
so that
$$
\psi_1  L_{|-1}= \psi \frac{\lambda_2(\mu_2+\mu_3)}{(\mu_{2}+\mu_3 -\lambda_2)^2}, \quad \psi_2  L_{|-2}= \psi \frac{\lambda_1(\mu_1+\mu_3)}{(\mu_{1}+\mu_3 -\lambda_1)^2},
$$
and, in view of \eqref{eq:mean},
$$
L = \frac{\rho}{1-\rho}+ \frac{\psi}{\mu(1-\rho)}\left(\frac{\lambda_2\mu_1(\mu_2+\mu_3)}{(\mu_{2}+\mu_3 -\lambda_2)^2}+\frac{\lambda_1\mu_2(\mu_1+\mu_3)}{(\mu_{1}+\mu_3 -\lambda_1)^2}\right).
$$

\section{Randomized Assignment}\label{sec:randomizeds}

We first apply our results to randomized load balancing schemes,
where each incoming job is assigned to a set of servers chosen at random,
independently of their current occupancy.
This oblivious load balancing
may cause a loss of performance compared to more sophisticated schemes,
but it has the advantage of involving no central authority to dispatch jobs.
As we will see, it is possible to leverage the symmetries of the system to compute performance metrics
with a complexity which is polynomial in the number of servers
while allowing for some heterogeneity. The  complexity of the recursive formulas presented in  this section and the following one, for randomized and local assignment schemes, respectively, are summarized in Table \ref{tab:complexity}.

\begin{table}[h]
	\begin{minipage}{\columnwidth}
		\begin{center}
			\begin{tabular}{|>{\centering}m{.14\textwidth}|>{\centering}m{.16\textwidth}|>{\centering}m{.2\textwidth}|>{\centering}m{.1\textwidth}|>{\centering}m{.1\textwidth}|>{\centering\arraybackslash}m{.1\textwidth}|}
				\hline 
				\multicolumn{3}{|c|}{Randomized assignment}
				& \multicolumn{3}{c|}{Local assignment} \\ 
				\hline 
				Hom. & Het. degrees & Het. pool & Nested & Hom. & Het. \\
				\hline
				$O(K)$ & $O(NK)$ & $O\left( N S K_1 \cdots K_S \right)$ 
				& $O(KI)$ 
				& $O(K^2)$ & $O(K^3)$ \\
				\hline
			\end{tabular}
			\caption{\label{tab:complexity} Complexities of the recursive formulas for different pool structures. \emph{Hom.} and \emph{Het.} stand for \emph{homogeneous} and \emph{heterogeneous}, respectively. 
			$ K $ is the number of servers, $ N $ the number of job types, $ S $ the number of server groups, $ K_s $ the number of servers in group $ s $, and $ I $ the number of job classes. Each entry gives the complexity to compute the global metric $ L $. In all structures but the homogeneous line, it is also the complexity to compute the metric $ L_i $ for a specific class $ i $. }
		\end{center}
	\end{minipage}
\end{table}

\subsection{Homogeneous pool}\label{sec:homogeneous}

We consider a pool of $K$ servers, each with service rate $\mu$.
Jobs arrive at rate $K \lambda$.
Upon arrival, each job is assigned to $d \le K$ servers chosen uniformly at random,
independently of the current state of the system, so that all jobs have the same degree of parallelism $ d $.
Since all servers are exchangeable, the load
$\rho = \lambda / \mu$
of the system is also the load of each server.
This model was considered in \cite{G17-1},
where is was shown that the system is stable if and only if $\rho < 1$.

\begin{figure}[ht]
  \centering
  \subfloat[Compact representation]{
    \begin{tikzpicture}[scale = 1.3]
      \node[server] (server1) at (-.8,0) {$\mu$};
      \node[server] (server2) {$\mu$};
      \node[server] (server3) at (.8,0) {$\mu$};

      \draw[type] (0,0) ellipse (1.4 and .6);
      \node at (1.1,.7) {$K=3$};

      \node[type] (job) at (0,2) {$K \lambda$};

      \draw[-] (job) edge node[midway,left] {$d=2$} (0,.6);
    \end{tikzpicture}
  }
  \qquad \qquad
  \subfloat[Expanded representation]{
    \begin{tikzpicture}[scale = 1.3]
      \node[server] (server1) at (-1,0) {$\mu$};
      \node[server] (server2) {$\mu$};
      \node[server] (server3) at (1,0) {$\mu$};

      \node[type] (job) at (0,2) {$K \lambda$};

      \node[class] (job1) at (-1,1) {$\frac{K \lambda}{\binom{K}d}$};
      \node[class] (job2) at (0,1) {$\frac{K \lambda}{\binom{K}d}$};
      \node[class] (job3) at (1,1) {$\frac{K \lambda}{\binom{K}d}$};

      \draw[densely dashed] (job) edge (job1);
      \draw[densely dashed] (job) edge (job2);
      \draw[densely dashed] (job) edge (job3);

      \draw[-] (job1) edge (server1);
      \draw[-] (job1) edge (server2);
      \draw[-] (job2) edge (server1);
      \draw[-] (job2) edge (server3);
      \draw[-] (job3) edge (server2);
      \draw[-] (job3) edge (server3);

      \draw[type,draw=none] (0,0) ellipse (1.4 and .6);
    \end{tikzpicture}
  }
  \caption{A homogeneous pool}
  \label{fig:random-homogeneous}
\end{figure}
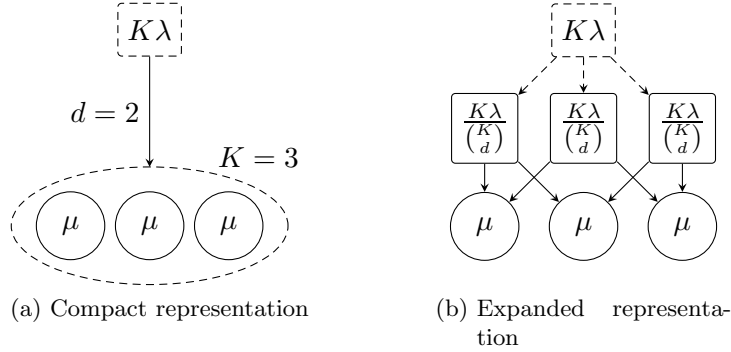

We will now apply Theorems \ref{theo:empty} and \ref{theo:mean}
to give a simple proof of the following results given in \cite[Theorems 1 and 2]{G17-1}:
\begin{equation}
  \label{eq:random-homogeneous}
  \psi = \prod_{\ell=d}^K \left( 1 - \rho_{|\ell} \right)
  \quad \text{and} \quad
  L = \sum_{\ell=d}^K \frac{ \rho_{|\ell} }{ 1 - \rho_{|\ell} },
\end{equation}
where
$$
\rho_{|\ell}
= \frac1{\ell \mu} \frac{ \binom{\ell}d }{ \binom{K}d } K \lambda
= \rho \frac{ \binom{\ell-1}{d-1} }{ \binom{K-1}{d-1} }
$$
denotes the load in the system restricted to $\ell$ arbitrary servers,
that is, the aggregate load generated by the job classes that can only be served by these $\ell$ servers.
These formulas can be evaluated with a complexity $O(K)$ if we compute the binomial coefficients by recursion as follows:
$$
\binom{\ell-1}{d-1}
= \left( 1 + \frac{d-1}{\ell-d} \right) \binom{\ell-2}{d-1},
\quad \forall \ell = d+1,\ldots,K,
$$
with the base case $\binom{d-1}{d-1} = 1$ when $\ell = d$.

In our framework, the class of a job defines the set of servers to which it was assigned upon arrival.
There are $I = \binom{K}d$ job classes,
one for each possible assignment of $d$ servers among $K$.
This is illustrated in Figure \ref{fig:random-homogeneous} with $K = 3$ servers
and a degree $d = 2$.
Since the assignment is uniform,
all classes have the same arrival rate $K \lambda / \binom{K}d$.
Thanks to the exchangeability of the servers,
we only need to keep track of the number of active servers
and not of their exact index when conditioning on the activity of the servers.

Specifically, for each $\ell = d,\ldots,K$, let $\psi_{|\ell}$ denote the probability that a system restricted to jobs processed by  $\ell$ arbitrary  servers is empty.
In this system,
the arrival rate is that of the jobs which are assigned to $d$ of these $\ell$ servers, namely
 $K \lambda \binom\ell{d} / \binom{K}d$.
The total service rate is $\ell \mu$.
Applying \eqref{eq:empty} then yields
$$
\psi_{|\ell}
= \left( 1 - \rho_{|\ell} \right)
\frac{ \ell \mu }{ \ell \frac\mu{\psi_{|\ell-1}} }
= \left( 1 - \rho_{|\ell} \right) \psi_{|\ell-1}.
$$
When there are $\ell = 1,\ldots,d-1$ servers left,
there are no more arrivals and the system is empty with probability $1$, i.e., $\psi_{|\ell} =1$.
The result announced for $\psi = \psi_{|K}$ follows by expanding the recursion.

Similarly, let $L_{|\ell}$ denote the mean number of jobs
in the system restricted to $\ell$  arbitrary servers, for each $\ell = 1,\ldots,K$.
\eqref{eq:mean} yields
\begin{align*}
  L_{|\ell}
  = \frac{\rho_{|\ell}}{1-\rho_{|\ell}}
  + \frac1{1 - \rho_{|\ell}} \frac
  { \ell \mu \frac{\psi_{|\ell}}{\psi_{|\ell-1}} L_{|\ell-1} }
  { \ell \mu }
  = \frac1{1 - \rho_{|\ell}} + L_{|\ell-1},
\end{align*}
for each $\ell = d,\ldots,K$, 
with the base cases $L_{|\ell} = 0$ for all $\ell = 1,\ldots,d-1$, from which \eqref{eq:random-homogeneous} follows.

Although the proof for $\psi$ is very close to that of \cite{G17-1},
the proof for $L$ is greatly simplified by Theorem \ref{theo:mean}.
We now see how to generalize the results to other classes of pools.

\subsection{Heterogeneous degrees}
\label{subsec:random-differentiated}

Consider a first extension where jobs can have different parallelism degrees.
There are still $K$ servers in the pool, each with service rate $\mu$,
but jobs are now divided into $N$ \emph{types}.
For each $u = 1,\ldots,N$,
type-$u$ jobs arrive
at rate $K \lambda p_u$,
with  $p_1 + \ldots + p_N = 1$, so that the total arrival rate is still $K\lambda$.
Upon arrival, a job of type $u$ is assigned to $d_u$ servers chosen uniformly at random,
independently of the current state of the system.
The load $\rho = \lambda / \mu$ of the system is also that of each server.
An  example with $N = 2$ job types  is given in Figure \ref{fig:random-differentiated}.

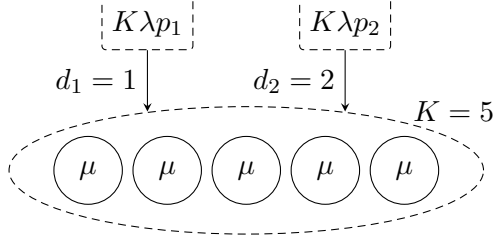
\begin{figure}[ht]
  \centering
  \begin{tikzpicture}[scale = 1.3]
    \node[server] (server1) at (-1.6,0) {$\mu$};
    \node[server] (server2) at (-.8,0) {$\mu$};
    \node[server] (server3) {$\mu$};
    \node[server] (server4) at (.8,0) {$\mu$};
    \node[server] (server5) at (1.6,0) {$\mu$};

    \draw[type] (0,0) ellipse (2.4 and .65);
    \node at (2.1,.6) {$K=5$};

    \node[type] (job1) at (-1,1.5) {$K \lambda p_1$};
    \node[type] (job2) at (1,1.5) {$K \lambda p_2$};

    \draw[-] (job1) edge node[midway,left] {$d_1=1$} (-1,.6);
    \draw[-] (job2) edge node[midway,left] {$d_2=2$} (1,.6);
  \end{tikzpicture}
  \caption{A homogeneous pool with two degree types.}
  \label{fig:random-differentiated}
\end{figure}

Using Theorems \ref{theo:empty} and \ref{theo:mean},
we can easily extend the results of the previous section.
For each $u = 1,\ldots,N$,
there are $\binom{K}{d_u}$ classes associated to type $u$,
one for each possible assignment of a type-$u$ job to $d_u$ servers among $K$.
All type-$u$ classes have the same arrival rate $K \lambda p_u / \binom{K}{d_u}$.
The exchangeability of the servers still ensures that we simply need to keep track of the number of active servers
when conditioning on their activity.

For each $\ell = 1,\ldots,K$,
let $\psi_{|\ell}$ denote the probability that the system restricted to $\ell$ arbitrary servers is empty.
For each $u = 1,\ldots,N$ with $d_u \ge \ell$,
there are $\binom\ell{d_u}$ type-$u$ classes
which are assigned to these $\ell$ servers,
so that the remaining arrival rate of type-$u$ jobs is
$K \lambda p_u \binom\ell{d_u} / \binom{K}{d_u}$.
For each $u = 1,\ldots,N$ with $d_u < \ell$,
there are no classes associated to type $u$ which are assigned to these $\ell$ servers only,
so that the arrival rate for type $u$ is zero.
In this case, we adopt the convention that $\binom\ell{d_u} = 0$, so that we can still write
$K \lambda p_u \binom\ell{d_u} / \binom{K}{d_u}$ for the arrival rate.
The total load in the system restricted to $\ell$ servers is then given by
\begin{equation*}
  \rho_{|\ell} = \frac1{ \ell \mu }
  \sum_{u=1}^N \frac{ \binom{\ell}{d_u} }{ \binom{K}{d_u} } K \lambda p_u
  = \rho \sum_{u=1}^N \frac{ \binom{\ell-1}{d_u-1} }{ \binom{K-1}{d_u-1} } p_u.
\end{equation*}
Observe that $\rho_{|\ell} < 1$ whenever $\rho < 1$
because $\binom{\ell}{d_u} / \binom{K}{d_u} \le \ell / K$ for each $u = 1,\ldots,N$.
Hence the system is stable whenever $\rho < 1$.
Using the exchangeability of the servers,
we can apply the same simplifications in \eqref{eq:empty} and \eqref{eq:mean} as in the homogeneous case,
so that $\psi$ and $L$ are still given by \eqref{eq:random-homogeneous}
where $\rho_{|\ell}$ is given by the expression above.
These formulas can be evaluated with a complexity $O(NK)$.
If a high number $R$ of values of the load $\rho$ is to be considered,
it is possible to precompute $\rho_{|\ell} / \rho$ for each
$\ell = 1,\ldots,K$
with complexity $O(NK)$
and then compute the results for each value of $\rho$ with complexity $O(RK)$,
so that the overall complexity is $O((N+R)K)$ instead of $O(RNK)$.

Since the jobs are differentiated by their degree, it is also interesting to evaluate the performance for each type of jobs individually.
It can be derived by applying \eqref{eq:meani} to  each class and then summing over all classes of the same type.
We obtain that the mean number of jobs of type $u$ in the system is given by
$$
L_u = \sum_{\ell=d_u}^K \frac{\rho_{u|\ell}}{1 - \rho_{u|\ell}},
$$
where, for each $\ell = d_u,\ldots,K$, $\rho_{u|\ell}$ is the load associated to type-$u$ jobs in the system with $\ell$ servers left:
$$
\rho_{u|\ell}
= \frac
{ \frac{ \binom{\ell}{d_u} }{ \binom{K}{d_u} } K \lambda p_u }
{ \ell \mu - \sum\limits_{\substack{v=1 \\ v \neq u}}^N \frac{ \binom{\ell}{d_v} }{ \binom{K}{d_v} } K \lambda p_v }
= \frac
{ \rho \frac{ \binom{\ell-1}{d_u-1} }{ \binom{K-1}{d_u-1} } p_u }
{ 1 - \rho \sum\limits_{\substack{v=1 \\ v \neq u}}^N \frac{ \binom{\ell-1}{d_v-1} }{ \binom{K-1}{d_v-1} } p_v } \text{.}
$$
The mean number of jobs of a given type can be evaluated with a complexity $O(NK)$.
It is again possible to make some precomputations when several values of the load $\rho$ are to be considered.

These results will be used in \S \ref{subsec:gain-of-differentiation}.

\subsection{Heterogeneous servers}
\label{subsec:random-heterogeneous}

We can further extend the model by considering server heterogeneity.
We distinguish $S$ groups of servers.
For each $s = 1,\ldots,S$, there are $K_s$ servers in group $s$, each with capacity $\mu_s$.
Like in \S \ref{subsec:random-differentiated}, we also distinguish $N$ types of jobs.
For each $u = 1,\ldots,N$, type-$u$ jobs arrive in the system at rate $K\lambda p_u$,
with $p_1 + \ldots + p_N = 1$. 
Upon arrival, each job of type $u$ is assigned to $d_{u,s}$ servers chosen uniformly at random among the $K_s$ servers of group $s$,
independently of the current state of the system, for each $s = 1,\ldots,S$.
The load of the system is now given by $\rho = K \lambda / \sum_{s=1}^S K_s \mu_s$.
Such a configuration is illustrated in Figure \ref{fig:random-heterogeneous}, with $N = 2$ types of jobs and $S = 2$ groups of servers.

\begin{figure}[ht]
  \centering
  \begin{tikzpicture}[scale = 1.3]
    \def\shift{.5}

    \node[server] (server11) at (-\shift-2.35,0) {$\mu_1$};
    \node at (-\shift-1.55,0) {$\ldots$};
    \node[server] (server12) at (-\shift-.75,0) {$\mu_1$};
    \node[server] (server21) at (\shift+.75,0) {$\mu_2$};
    \node at (\shift+1.55,0) {$\ldots$};
    \node[server] (server22) at (\shift+2.35,0) {$\mu_2$};

    \draw[type] (-\shift-1.55,0) ellipse (1.4 and .6);
    \node at (-\shift-2.75,.55) {$K_1$};
    \draw[type] (\shift+1.55,0) ellipse (1.4 and .6);
    \node at (\shift+2.75,.55) {$K_2$};

    \node[type] (job1) at (0,1.3) {Type-$1$ jobs};
    \node[type] (job2) at (0,-1.3) {Type-$2$ jobs};

    \draw (job1) edge node[midway,left,yshift=3] {$d_{1,1}$} (-\shift-1,.55);
    \draw (job1) edge node[midway,right,yshift=3] {$d_{1,2}$} (\shift+1,.55);
    \draw (job2) edge node[midway,left,yshift=-3] {$d_{2,1}$} (-\shift-1,-.55);
    \draw (job2) edge node[midway,right,yshift=-3] {$d_{2,2}$} (\shift+1,-.55);
  \end{tikzpicture}
  \caption{A heterogeneous pool with two degree types.}
  \label{fig:random-heterogeneous}
\end{figure}
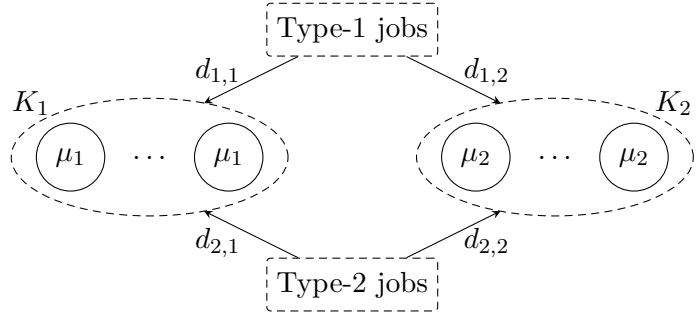

We now apply our framework to this heterogeneous pool.
For each $u = 1,\ldots,N$,
a class associated to type $u$ is defined by choosing independently $d_{u,s}$ servers within group $s$, for each $s = 1,\ldots,S$.
Thus there are $\prod_{s=1}^S \binom{K_s}{d_{u,s}}$ classes associated to type-$u$ jobs,
each with the same arrival rate.
Since the servers from different groups are not exchangeable,
we need to keep track of the number of servers \emph{within each group} when conditioning on their activity.

For each $\ell = (\ell_s: s = 1,\ldots,S)$, with $\ell_s \le K_s$ for each $s = 1,\ldots,S$,
we let $\psi_{|\ell}$ denote the probability that the system restricted to $\ell_s$ arbitrary servers of group $s$ for each $s = 1,\ldots,S$ is empty.
We also let $L_{|\ell}$ denote the mean number of jobs in this system.
By an argument similar to those of the previous sections, we obtain that the load in this restricted system is
$$
\rho_{|\ell} = \rho
\left( \sum\limits_{u=1}^N p_u \prod\limits_{s=1}^S \frac{ \binom{\ell_s}{d_{u,s}} }{ \binom{K_s}{d_{u,s}} } \right)
\frac{ \sum_{s=1}^S K_s \mu_s }{ \sum_{s=1}^S \ell_s \mu_s }.
$$
Accounting for the server exchangeability within each group in \eqref{eq:empty} and \eqref{eq:mean} yields directly
\begin{align*}
  \psi_{|\ell}
  &= \left( 1 - \rho_{|\ell} \right) \frac
  { \sum_{s=1}^S \ell_s \mu_s }
  { \sum_{s=1}^S \ell_s \frac{\mu_s}{\psi_{|\ell-e_s}} }
  \quad \text{and} \quad
  L_{|\ell}
  = \frac{\rho_{|\ell}}{1 - \rho_{|\ell}}
  + \frac1{1 - \rho_{|\ell}} \frac
  { \sum_{s=1}^S \ell_s \mu_s \frac{\psi_{|\ell}}{\psi_{|\ell-e_s}} L_{|\ell-e_s} }
  { \sum_{s=1}^S \ell_s \mu_s }.
\end{align*}
Hence we can compute $\psi$ and $L$ by recursion,
with complexity $O(N S K_1 \cdots K_S)$,
which is $ O(NS(\frac{K}{S})^S) $ in the worst case.
While still polynomial in $ K $, the complexity suggests to limit the study to small values of $ S $.
If a high number $R$ of values of the load $\rho$ is to be considered,
it is possible to precompute $\rho_{|\ell} / \rho$ for each $\ell$
with complexity $O(N S K_1 \cdots K_S)$
and then compute the results for each value of $\rho$ with complexity $O(R S K_1 \cdots K_S)$,
so that the overall complexity is $O((N+R) S K_1 \cdots K_S)$ instead of $O(R N S K_1 \cdots K_S)$.

Similarly, applying \eqref{eq:meani} per class and then summing over all classes associated to the same type
give the following recursion for the mean number of jobs of type $u$, for each $u = 1,\ldots,N$:
$$
L_{u|\ell} =
\frac{ \rho_{u|\ell} }{ 1 - \rho_{u|\ell} }
+ \frac1{ 1 - \rho_{|\ell} } \frac
{ \sum_{s=1}^S \ell_s \mu_s \frac{\psi_{|\ell}}{\psi_{|\ell-1}} L_{u|\ell-e_s} }
{ \sum_{s=1}^S \ell_s \mu_s },
$$
where $\rho_{u|\ell}$ is the load associated to type-$u$ jobs in the system with $\ell_s$ servers left in group $s$, for each $s = 1,\ldots,S$, given by
$$
\rho_{u|\ell} = \frac
{ \rho
  \left( \prod\limits_{s=1}^S \frac{ \binom{\ell_s}{d_{u,s}} }{ \binom{K_s}{d_{u,s}} } \right) p_u
  \frac{ \sum_{s=1}^S K_s \mu_s }{ \sum_{s=1}^S \ell_s \mu_s }
}
{ 1 - \rho
  \left( \sum\limits_{\substack{v=1 \\ v \neq u}}^N \prod\limits_{s=1}^S \frac{ \binom{\ell_s}{d_{v,s}} }{ \binom{K_s}{d_{v,s}} } p_v \right)
  \frac{ \sum_{s=1}^S K_s \mu_s }{ \sum_{s=1}^S \ell_s \mu_s }
} \text{.}
$$
It is again possible to make some precomputations
when several values of the load $\rho$ are to be considered.

\section{Local Assignment}\label{sec:ranged-clusters}

In the previous section, we have assumed that a job could be assigned to an arbitrary subset of servers. This large degree of freedom can be difficult to implement in practice. For example, one may want to select physically-close servers in order to minimize the communication overhead. This is what we call \emph{local} assignment.

In this section, we abstract the concept of localization by introducing \emph{line pools}, where servers are assumed to be located along a line and indexed by the integers $1,\ldots,K $ so that servers $i$ and $j$ are at physical distance $|i-j|$. 
The locality constraint is modeled by the assignment graph: each job class is assigned to an integer interval. For ease of notation, we identify a class and its assigned range:  $ i,j $ denotes the class that is assigned servers $i$ to $j$. An illustration of a line pool is given in Figure \ref{fig:range}.

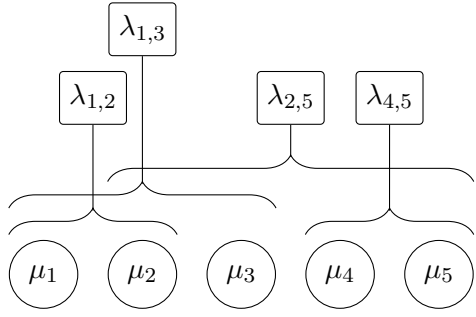
\begin{figure}[h]
  \centering
  \begin{tikzpicture}[scale = 1.3, draw, node distance=2cm,auto,>=stealth]
    \node[class] (l13) at (-1,2.5) {$\lambda_{1,3}$};
    \node[class] (l25) at (.5,1.8) {$\lambda_{2,5}$};
    \node[class] (l12) at (-1.5,1.8) {$\lambda_{1,2}$};
    \node[class] (l45) at (1.5,1.8) {$\lambda_{4,5}$};

    \node[server] (u1) at (-2,0) {$ \mu_1 $};
    \node[server] (u2) at (-1,0) {$ \mu_2 $};
    \node[server] (u3) at (0,0) {$ \mu_3 $};
    \node[server] (u4) at (1,0) {$ \mu_4 $};
    \node[server] (u5) at (2,0) {$ \mu_5 $};

    \draw [decorate,decoration={brace,raise=15pt, amplitude = 10pt},yshift=0pt]
    (u1.west) -- (u2.east);
    \draw [decorate,decoration={brace,raise=15pt, amplitude = 10pt},yshift=0pt]
    (u4.west) -- (u5.east);
    \draw [decorate,decoration={brace,raise=25pt, amplitude = 10pt},yshift=0pt]
    (u1.west) -- (u3.east);
    \draw [decorate,decoration={brace,raise=35pt, amplitude = 10pt},yshift=0pt]
    (u2.west) -- (u5.east);

    \draw (l12) -- (-1.5,.65)
    (l45) -- (1.5,.65)
    (l13) -- (-1,.925)
    (l25) -- (.5,1.2);

  \end{tikzpicture}
  \caption{A line pool}
  \label{fig:range}
\end{figure}

The rest of the section is organized as follows. We first introduce nested pools, a special type of line pools introduced in \cite{G17-2}. We then study arbitrary line pools, followed by a local version of the randomized schemes studied above. We conclude the section by giving the behavior for a \emph{ring} structure.

\subsection{Nested structures}
\label{sebsec:nested}

A pool is said to be nested if the following property is verified: 
$$
\forall i,j\in {\mathcal I},\quad {\mathcal K}_i \cap {\mathcal K}_j \ne \emptyset \quad \Longrightarrow \quad {\mathcal K}_i \subset {\mathcal K}_j \quad \text{or} \quad {\mathcal K}_j \subset {\mathcal K}_i.
$$
Thus, if two jobs share a server, then the servers assigned to one of these jobs form a subset of the servers assigned to the other job. Without loss of generality, we can always assume that class $ 1,K $, which is assigned to all servers, exists. Otherwise, as observed in \cite{G17-2}, we can split the pool into smaller, independent, nested pools, and consider each sub-pool separately. 
While a line pool is not necessarily a nested pool (consider for example classes $ 1,3 $ and $ 2,5 $ in Figure \ref{fig:range}), the converse holds:

\begin{prop}
	\label{prop:neste_is_range}
	A nested pool is a line pool.
\end{prop}

\begin{proof}
  We first remark that a nested pool has a natural tree structure, which can be built as follows. The nodes are the servers and the  job classes. The parent of a server is the smallest class, in the sense of inclusion, assigned to it. The parent of a class is the smallest class that includes it, if any. By construction, servers are always leaves, while the tree root is exactly the class that is maximal for the inclusion, i.e., class  $ 1,K $.

  To conclude, we just have to label the servers in their order of appearance in a depth-first traversal of the tree. By construction, the servers assigned to a given class, which are exactly the leaves of the subtree rooted in that class, will have consecutive labels.
\end{proof}

An example of nested pool is given in Figure \ref{fig:nested}.

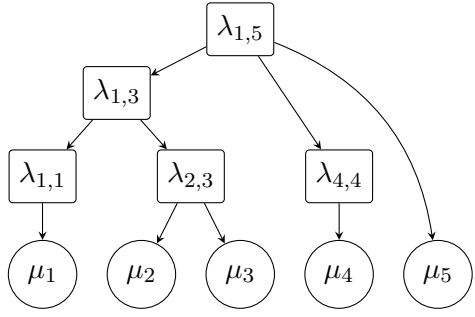
\begin{figure}[h]
  \centering
  \begin{tikzpicture}[scale = 1.3, draw, node distance=2cm,auto,>=stealth]
    \node[class] (lr) at (0,2.5) {$\lambda_{1,5}$};
    \node[class] (la) at (-1.25,1.85) {$\lambda_{1,3}$};
    \node[class] (lb) at (-2,1) {$\lambda_{1,1}$};
    \node[class] (lc) at (-.5,1) {$\lambda_{2,3}$};
    \node[class] (ld) at (1,1) {$\lambda_{4,4}$};	

    \node[server] (u1) at (-2,0) {$ \mu_1 $};
    \node[server] (u2) at (-1,0) {$ \mu_2 $};
    \node[server] (u3) at (0,0) {$ \mu_3 $};
    \node[server] (u4) at (1,0) {$ \mu_4 $};
    \node[server] (u5) at (2,0) {$ \mu_5 $};

    \draw (lr) edge (la) edge (ld) edge[bend left] (u5)
    (la) edge (lb) edge (lc)
    (lb) edge (u1)
    (lc) edge (u2) edge (u3)
    (ld) edge (u4);
  \end{tikzpicture}
  \caption{Tree representation of a nested pool}
  \label{fig:nested}
\end{figure}

Nested pools 
are another good example of application of our recursive formula.
It was shown in \cite{G17-2} that a nested system is empty with probability
\begin{equation}
\label{eq:psi_nested}
\psi = \prod_{i \in \mathcal{I}} (1 - \rho_{|i})\text{,}
\end{equation}
\noindent where
$$
\rho_{|i} = \frac{\lambda_i}{\sum_{k \in \mathcal{K}_i} \mu_k - \sum_{j : \mathcal{K}_j \subsetneq \mathcal{K}_i} \lambda_j}
$$
is the load associated to class $i$
 in the system restricted to  servers $\mathcal{K}_i$.
With our recursive  approach, proving \eqref{eq:psi_nested} becomes quite straightforward.
Using \eqref{eq:emptyderivative}, we get
\begin{align*}
{\psi} &=\frac
  { M(\mathcal{K}) - \Lambda(\mathcal{I}) }
  { \sum_{k \in \mathcal{K}} \frac{\mu_k}{\psi_{|-k}} }
=({1-\rho_{|1,K}}) \frac
{ M(\mathcal{K}) - \Lambda(\mathcal{I}\setminus \{1,K\}) }
{ \sum_{k \in \mathcal{K}} \frac{\mu_k}{\psi_{|-k}} }.
\end{align*}
We then remark that if any server $ k $ is idle, so is class $ 1,K $. Hence the second factor of the expression above is exactly the right-hand side of \eqref{eq:emptyderivative} for a system where class $ 1,K $ is removed, that is $ \psi_{|-1,K} $. Thus, $\psi= (1-\rho_{|1,K})\psi_{|-1,K}$, from which  \eqref{eq:psi_nested} follows.

Note that \eqref{eq:psi_nested} can also be proved with a more class-oriented approach.
Conditioning on the activity of class $1,K$, we get the conservation equation:
$$
\sum_{i \in \mathcal{I}} \lambda_i
= \left( \sum_{k \in \mathcal{K}} \mu_k \right) (1 - \psi_{1,K}) + \left( \sum_{i \in \mathcal{I} \setminus \{1,K\}} \lambda_i \right) \psi_{1,K}\text{.}
$$  
Rearranging the terms gives:
$$
\psi_{1,K} = \frac
{ M(\mathcal{K}) - \Lambda(\mathcal{I}) }{M(\mathcal{K}) - \Lambda(\mathcal{I}\setminus\{1,K\})}
= 1 - \rho_{|1,K}.
$$
The result then follows from 
the equality $\psi = \psi_{1,K} \psi_{|-1,K}$.

These proofs give some insight on the factors in \eqref{eq:psi_nested}. For example, we see that the equality $\psi_{i,j} = 1-\rho_{|i,j}$ is only true when $i,j= 1,K$. Indeed, the proof consists in \emph{removing} the classes one after the other in a graph traversal, showing that $1 - \rho_{|i,j}$ is the probability that class $i,j$ is idle,
given that all its ancestors in the tree (if any) are idle. 

The mean number of jobs of each class, which was also given in \cite{G17-2},
can be derived using \eqref{eq:meani_s}.
It is a special case of
\eqref{eq:range_single_class_v2}
that will be stated for the line pools.

\subsection{Line structures}

We now remove the nested assumption and show how to apply the recursive formula to any line pool. 

\begin{prop}
  \label{prop:arbitrary_ranges}
  The probability that a line system is empty is given by
  \begin{equation}
    \label{eq:range_norm_v2}
    \psi=\frac{M(\mathcal{K})-\Lambda(\mathcal{I})}{\sum_{k=1}^K\frac{\mu_k}{\psi_{|1..k-1}\psi_{|k+1..K}}},
  \end{equation}
  where for each $k,\ell$, $|k..\ell$ denotes the system reduced to servers $k$ to $\ell$.

  For each $i,j \in \mathcal{I}$, the mean number of class-$i,j$ jobs satisfies
  \begin{equation}
    \label{eq:range_single_class_v2}
    L_{i,j} =
    \frac{\lambda_{i,j}+\psi\left(
    \sum\limits_{k = 1}^{i-1}\mu_k\frac{L_{{i,j}|k+1..K}}{\psi_{|1..k-1}\psi_{|k+1..K}}
    +
    \sum\limits_{k = j+1}^K\mu_k\frac{L_{{i,j}|1..k-1}}{\psi_{|1..k-1}\psi_{|k+1..K}}
    \right)
    }{M(\mathcal{K})-\Lambda(\mathcal{I})}\text{,}
  \end{equation}
  while the total mean number of jobs satisfies
  \begin{equation}
    \label{eq:range_total_mean_v2}
    L =
    \frac{\Lambda(\mathcal{I}) + \psi \sum\limits_{k=1}^{K}\mu_k\frac{L_{|1..k-1}+L_{|k+1..K}}{\psi_{|1..k-1}\psi_{|k+1..K}}
    }{M(\mathcal{K})-\Lambda(\mathcal{I})}\text{.}
  \end{equation}
  \noindent 
\end{prop}

\begin{proof}
  The key of the proof is that when we remove some server $ k $ from a line pool, we get two independent line pools, in the sense that the remaining classes $ \mathcal{I} \setminus \mathcal{I}_k $ are split into two sets: those processed by servers $1$ to $k-1$ and those processed by servers $k+1$ to $K$. This yields $ \psi_{|-k} = \psi_{|1\ldots k-1} \psi_{|k+1 \ldots K}$. Equation \eqref{eq:range_norm_v2} then follows from \eqref{eq:emptyderivative}. 

  For the mean number of jobs, we have
  $$
  L_{{i,j}|-k} = 
  \left\{\begin{array}{ll}
    L_{{i,j}|k+1..K} & \text{if } k<i,\\
    L_{{i,j}|1..k-1} & \text{if } k>j,\\
    0 & \text{otherwise,}
  \end{array}\right.
  $$
  and 
  $ L_{|-k} = L_{|1..k-1}+L_{|k+1..K} $.
  The recursive formulas \eqref{eq:range_single_class_v2} and \eqref{eq:range_total_mean_v2} then follow from \eqref{eq:meani_s} and \eqref{eq:mean_s}, respectively.
\end{proof}

In view of Proposition \ref{prop:arbitrary_ranges}, the computation of $\psi$ can be done in time $O(K^3)$. First,  we precompute the total arrival and service rates of all pools reduced to servers $k$ to $\ell$, for all $k$ and $\ell$ such that $1\le k \le \ell \le K$, which incurs a cost in $O(K^2)$.
 Then,  the computational cost of each term $ \psi_{|k..\ell} $ is in $ O(K) $, and there are $ O(K^2) $ such terms, hence a global cost $ O(K^3) $. Keeping the different values of 
$ \psi_{|k..\ell} $ in memory, the same complexity argument  holds for \eqref{eq:range_single_class_v2} and \eqref{eq:range_total_mean_v2}.
The mean number of jobs of any class $ {i,j} $ and the total number of jobs can be computed in time $ O(K^3) $.

Note that Proposition \ref{prop:neste_is_range} ensures that the above recursive formulas also apply to nested pools. However, the equations derived in \S\ref{sebsec:nested} for nested pools are simpler to compute: for example, using the tree structure of the classes, one can verify that the computational cost of \eqref{eq:psi_nested} is $ O(IK) $, against $ O(K^3) $ for a generic line pool. Reminding that $ I=O(K^2) $, with $I$ possibly much lower than $K^2$, nested formulas should be preferred when the pool is nested.

It is tempting to adapt the method presented here for line pools to other topologies. For example, one could consider a grid structure where job classes would correspond to rectangles of servers. However, the method does not apply, as removing a server does not  yield to  independent sub-systems in general. A notable exception, considered in \S\ref{sec:ring-structure}, is the ring topology.

\subsection{Load balancing}\label{sec:randomized-load-balancing}

Section \ref{sec:randomizeds} investigated randomized assignment where a fixed number of servers were chosen at random in a pool. We show here that these results can be transposed to line pools. We only consider the homogeneous case; the cases of  heterogeneous degrees or servers can be treated in the same way.

As in \S \ref{sec:homogeneous}, we consider $ K $ servers, each with service rate $ \mu$. Jobs arrive at rate $ K\lambda $, so that the system load is $\rho = \lambda / \mu$. Upon arrival, each job is assigned to a range of $ d\leq K $ servers chosen uniformly at random among the $ I = K-d+1 $ possible ranges of size $ d $, so that the arrival rate for each of the $ I $ classes is $ \frac{K\lambda}{K-d+1} $. An example is pictured in Figure \ref{fig:rangebalancing}. As $ d $ is constant, we label a class by its lowest server, i.e., $ i $ instead of $ {i,i+d-1} $. Remark that, unlike what happens for the non-local case considered in \S \ref{sec:homogeneous}, the classes are not equivalent. For example, in Figure \ref{fig:rangebalancing}, class $ 1 $ has an exclusive use of server $ 1 $, while class $ 3 $ shares its three  servers. 

\begin{figure}[h]
  \centering
  \begin{tikzpicture}[scale = 1.3, draw, node distance=2cm,auto,>=stealth]
    \def\K{6}
    \pgfmathsetmacro\res{(\K+1)/2}
    \node[type] (toto) at (1.2*\res,2) {$ K\lambda $};
    \foreach \n in {1,...,\K}
    \node[server] (u\n) at (1.2*\n,0) {$ \mu $};
    \pgfmathtruncatemacro\res{\K-1}
    \foreach \n in {2,...,\res}{
      \node[class] (l\n) at (1.2*\n,1) {$ \frac{\lambda}{1-\frac{d-1}{K}}$};
    \draw (toto) edge[densely dashed] (l\n);
    \foreach \i in {-1,0,1}
    \pgfmathtruncatemacro\res{\n+\i}
    \draw (l\n) edge (u\res);
  }
  \end{tikzpicture}
	\caption{A homogeneous line pool}
	\label{fig:rangebalancing}
\end{figure}
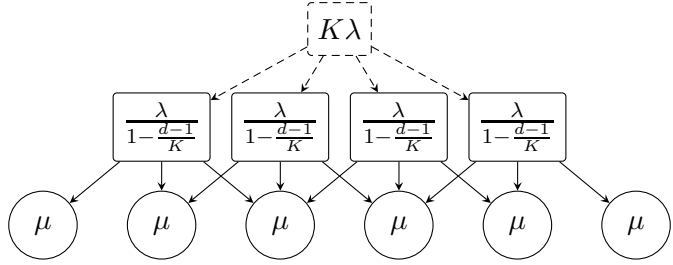

\begin{prop}
	\label{prop:range_homo}
	The probability that the system is empty is $\psi = \psi_{|1..K}$, where 
	$\psi_{|1..\ell}$ can be computed recursively by $\psi_{|1..\ell} = 1$ for $\ell < d$ and 
	\begin{equation}
	\label{eq:empty_range_homo}
	\psi_{|1..\ell} = 
	\frac{1-\rho_{|1..\ell}}{\frac 1 \ell \sum_{k=1}^{\ell}\frac{1}{\psi_{|1..k-1}\psi_{|1..\ell-k}}}
	\text{,}
  \end{equation}
  for $\ell \ge d$, where 
  $$
  \rho_{|1..\ell} =  \frac{1-\frac{d-1}{\ell}}{1-\frac{d-1}{K}}\rho.
  $$
	
	For each $i\in 1, \ldots, I$, the mean number of class-$i$ jobs  is $L_{i} = L_{i|1..K}$, where 
	$L_{i|1..\ell}$ can be computed recursively by $L_{i|1..\ell} = 0$ for $\ell < d$ and
	\begin{equation}
	\label{eq:single_range_homo}
	L_{i|1..\ell} =
	\frac{\frac{\rho_{|1..K}}{1-\frac{d-1}{K}}		
		+\psi_{|1..\ell}\left(
		\sum\limits_{k=1}^{i-1}\frac{L_{i-k|1..\ell-k}}{\psi_{|1..k-1}\psi_{|1..\ell-k}}
		+
		\sum\limits_{k=i+d}^{\ell}\frac{L_{i|1..k-1}}{\psi_{|1..k-1}\psi_{|1..\ell-k}}
		\right)
    }{\ell(1-\rho_{|1..\ell})}
  \end{equation} 
  for $\ell \ge d$.

  The total mean number of jobs in the system is $L = L_{|1..K}$, with $L_{|1..\ell} = 0$ if $\ell < d$ and 
  \begin{equation}
    \label{eq:total_range_homo}
    L_{|1..\ell} = \frac{
      \rho_{|1..\ell}+\frac{\psi_{|1..\ell}}{\ell}\sum_{k=1}^{\ell}\frac{L_{|1..k-1}+L_{|1..\ell-k}}{\psi_{|1..k-1}\psi_{|1..\ell-k}}
    }{1-\rho_{|1..\ell}
    }
  \end{equation}
  otherwise.
\end{prop}

\begin{proof}
  The result follows from Proposition \ref{prop:arbitrary_ranges} on observing that the pool restricted to servers $\ell +1$ to $K$ is equivalent to the pool restricted to servers $1$ to $K-\ell$, for any $\ell  < K$.
\end{proof}

The recursions \eqref{eq:empty_range_homo} and \eqref{eq:total_range_homo} use $ O(K) $ values of $ \psi $ and $ L $, and each of them can be computed in $ O(K) $ if previous results are kept in memory, resulting in $ O(K^2) $ time complexity. For \eqref{eq:single_range_homo}, there are $ O(K^2) $ values to compute, hence a computational cost in $ O(K^3) $.
Despite important symmetries, there is no complexity gain for per-class performance compared to the general case, because classes remain heterogeneous. However, an improvement of factor $ K $ is achieved for the global indicators $\psi$ and $L$.

It is worth noting that, despite heterogeneity, the stability condition is simply $ \rho<1 $. The reason is that the sub-systems are less loaded than the main system: the sub-system restricted to servers $1$ to $\ell$ has load $\rho_{|1..\ell}\le \rho$.

\subsection{Ring structure}\label{sec:ring-structure}

To suppress the class asymmetry inherent to line pools, we now consider a ring pool where servers 1 and $K$ are at distance 1, as illustrated  in Figure \ref{fig:ring}. To simplify formulas, we use implicit congruence modulo $ K $: server $ K+i $ is server $ i $, and for $1\leq j<i\leq K $, $ i,\ldots,j $ denotes the servers $i$ to $K$ and $1$ to $j$. For example, in Figure \ref{fig:ring}, class-$ {5,2} $ jobs are assigned servers 5, 1 and 2.

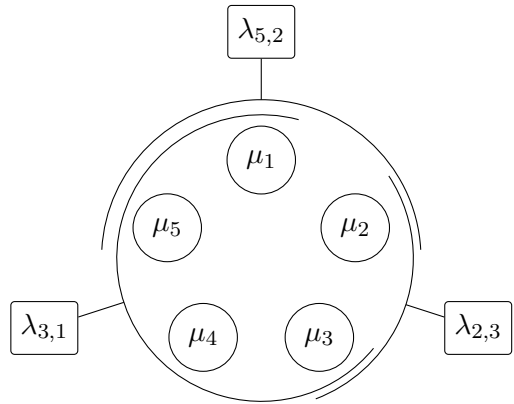
\begin{figure}[h]
  \centering
  \begin{tikzpicture}[scale = 1,draw, node distance=2cm,auto,>=stealth]

    \foreach \n/\a in {1/90,2/18,3/306,4/234,5/162}
    \node[server] (u\n) at (\a:1.3) {$ \mu_{\n} $};

    \node[class] (l24) at (342:3) {$ \lambda_{{2,3}} $};
    \draw[shift = (33:2)] (0:0) arc (33:-69:2) ;
    \draw (l24) -- (342:2);

    \node[class] (l41) at (-162:3) {$ \lambda_{{3,1}} $};
    \draw[shift = (75:1.9)] (0:0) arc (75:321:1.9) ;
    \draw (l41) -- (-162:1.9);

    \node[class] (l52) at (90:3) {$ \lambda_{{5,2}} $};
    \draw[shift = (3:2.1)] (0:0) arc (3:177:2.1) ;
    \draw (l52) -- (90:2.1);

  \end{tikzpicture}
  \caption{A ring pool}
  \label{fig:ring}
\end{figure}

The following result is a  simple consequence of  Proposition \ref{prop:arbitrary_ranges}, after noticing that removing server $ k $ from a ring gives the line pool $ k+1,\ldots,k-1 $.

\begin{prop}
  \label{theo:arbitrary_rings}
  The probability that a ring system is empty is given by
  \begin{equation}
    \label{eq:ring_empty}
    {\psi} = \frac{M(\mathcal{K})-\Lambda(\mathcal{I})}{\sum_{k=1}^K\frac{\mu_k}{\psi_{|k+1..k-1}}}\text{,}
  \end{equation}
  \noindent where $\psi_{|k+1..k-1}$ is obtained using \eqref{eq:range_norm_v2}.

  For each ${i,j} \in \mathcal{I}$, the mean number of class-${i,j}$ jobs in the system is given by
  \begin{equation}
    \label{eq:ring_single_class}
    {L}_{{i,j}} =
    \frac{\lambda_{{i,j}}+ {\psi}
    \sum_{k\in j+1..i-1}\mu_k\frac{L_{{i,j}|k+1..k-1}}{\psi_{|k+1..k-1}}
    }{M(\mathcal{K})-\Lambda(\mathcal{I})}\text{,}
  \end{equation}
  \noindent  and the total mean number of jobs in the system is
  \begin{equation}
    \label{eq:ring_total_mean}
    {L} =
    \frac{\Lambda(\mathcal{I}) +  {\psi} \sum_{k=1}^{K}\mu_k\frac{L_{|k+1..k-1}}{\psi_{|k+1..k-1}}
    }{M(\mathcal{K})-\Lambda(\mathcal{I})}\text{,}
  \end{equation}
  \noindent 
  where $L_{{i,j}|k+1..k-1}$ and $L_{|k+1..k-1}$ are obtained using \eqref{eq:range_single_class_v2} and \eqref{eq:range_total_mean_v2}. 
\end{prop}

The complexity of each of these recursions is in $O(K^3)$.
For a homogeneous ring of load $\rho$ and range size $d < K$, the complexity is reduced to $O(K^2)$.
In this case, all classes are equivalent and we only need to focus on the global metrics $\psi$ and $L$, given by
\begin{equation}
\label{eq:ring-homogeneous}
 {\psi} = (1-\rho) \psi_{|1..K-1}
\quad \text{and} \quad
 {L} = \frac{\rho}{1-\rho}+L_{|1..K-1},
\end{equation}
where $ \psi_{|1..K-1} $ and $ L_{|1..K-1} $ are the metrics associated with the line system restricted to servers $1,\ldots,K-1$, with load
$$
\rho_{|1..K-1}=  \rho \left(1-\frac{d-1}{K-1}\right).
$$

The ring topology is commonly used in  Distributed Hash Tables (DHT)  to access resources in a decentralized fashion, as in  the Chord protocol \cite{chord}. In the ring space, a portion is a connected subset. We can easily imagine how to use a DHT to dispatch jobs to a pool of servers: it would be enough to let the DHT index the servers; when a job enters the system, it contacts one indexing node that returns the set of servers that are mapped to its monitoring area. The resulting pool would behave exactly like a ring pool.

\section{Numerical Evaluation}\label{sec:numerical-evaluation}

We now illustrate the previous results through  two studies: relevance of the parallelism degree to achieve implicit service differentiation, and performance degradation due to localized load balancing.
Observe that, given the complexity of the involved performance metrics, these studies would not have been possible without our recursive formulas.

For easing the display of the results, the performance of class-$ i $ jobs is quantified by the inverse of the mean response time, $ \frac{1}{T_i} $. This metric also happens to be the mean service rate $\gamma_i$ received by the jobs of class $i$, as we have
\begin{align*}
  \gamma_i
  = \frac{\sum_x \pi(x) \phi_i(x)}{\sum_x \pi(x) x_i}
  = \frac{\lambda_i}{\sum_x \pi(x) x_i}
  = \frac{\lambda_i}{L_i} = \frac{1}{T_i},
\end{align*}
where the second equality holds by the conservation equation.

\subsection{Gain of differentiation}
\label{subsec:gain-of-differentiation}

Consider a resource pool with two types of jobs called \emph{regular} and \emph{premium}.
A natural way of differentiating services consists in
assigning premium jobs to more servers than regular jobs.
We are interested in assessing the actual impact of this approach on  performance.

For the numerical results, we consider $K = 100$ servers with unit service rates;
regular jobs have a parallelism degree $6$ while premium jobs have a degree $12$.
This corresponds to the model of \S \ref{subsec:random-differentiated},
with $N = 2$ job types, regular and premium.
We first focus on the influence of load on the efficiency of the service differentiation.

\paragraph{Impact of load}

Figure \ref{fig:diff1} shows the mean service rates
as a function of the system load $\rho$, for three population distributions:
regular jobs only, premium jobs only,
and a mixed population where regular and premium jobs generate half of the load.

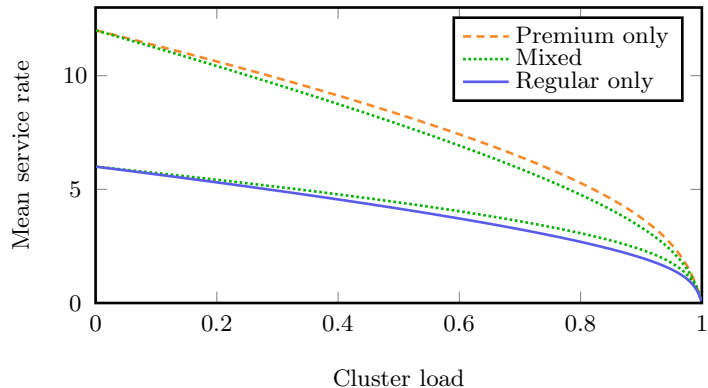
\begin{figure}[h]
\centering
\pgfplotstableread[col sep=comma]{hetero1_100.txt}\t
\begin{tikzpicture}
\begin{axis}[
tikzDefaults,
xlabel={Cluster load}, ylabel={Mean service rate},
xmin=0, xmax=1, ymax = 13,
legend pos = north east,
legend style={
	row sep=-0.4cm,
	font = \footnotesize,
	inner ysep=-.1cm,
	cells={anchor=west, align=left},
},
width=.65\linewidth, height=5.5cm,
]

\addplot[myorange,densely dashed] table [x = rho, y expr = 1/\thisrow{d12}]{\t};
\addlegendentry{Premium only}

\addplot[green!70!black,densely dotted] table [x = rho, y expr = 1/\thisrow{d6m}]{\t};
\addplot[green!70!black,densely dotted,forget plot] table [x = rho, y expr = 1/\thisrow{d12m}]{\t};
\addlegendentry{Mixed}

\addplot[myblue,solid] table [x = rho, y expr = 1/\thisrow{d6}]{\t};
\addlegendentry{Regular only}

\end{axis}
\end{tikzpicture}
\caption{Impact of load on service differentiation for different populations.
  Top plots give the performance of premium jobs and bottom plots that of the regular jobs.}
  \label{fig:diff1}
\end{figure}

The service qualities of the two types are clearly different.
When the load is low, the service rate of premium jobs is roughly twice that of regular jobs.
Intuitively, if the arrivals are rare, then it is likely that a new job finds all its servers free upon arrival.
The ratio between the service rates of premium and regular jobs decreases with the load but remains significant until the load is extreme.
Premium and regular jobs seem to have asymptotically the same service rate as $ \rho $ tends to 1.
This convergence is somehow expected, as maintaining a minimal ratio greater than 1 at very high load
could jeopardize the stability of the system for regular jobs.

Interestingly, the service rate of premium jobs is lower when half of the population consists of regular jobs.
The reason is that the slowness of regular jobs penalizes premium jobs, as they stay longer in the system.
This also explains the gain of performance for regular jobs when population is mixed.
This is particularly visible at higher load, when the job interactions intensify.

\paragraph{Impact of population distribution}

Following up with the last observation,
we focus on the impact of the proportion of regular and premium jobs in the population.
Since we observed that this impact is stronger when the load is high,
Figure \ref{fig:diff2} gives the mean service rate under loads $\rho = 0.9$ and $\rho = 0.99$,
as a function of the ratio of the arrival rate of regular jobs to that of premium jobs.

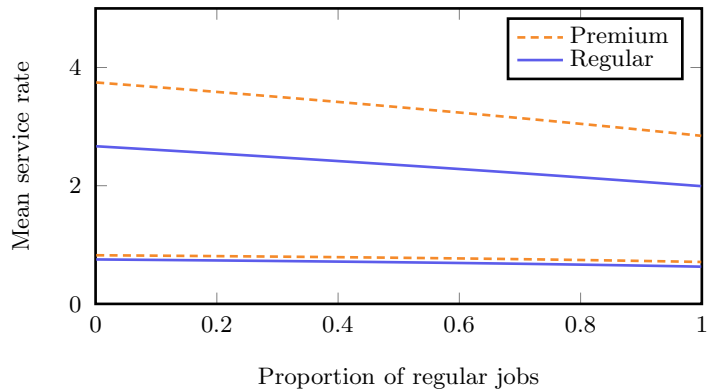
\begin{figure}[h]
\centering
\pgfplotstableread[col sep=comma]{hetero2.txt}\t
\begin{tikzpicture}
\begin{axis}[
tikzDefaults,
xlabel={Proportion of regular jobs}, ylabel={Mean service rate},
xmin=0, xmax=1, ymax = 5,
legend pos = north east,
legend style={
	row sep=-0.4cm,
	font = \footnotesize,
	inner ysep=-.1cm,
	cells={anchor=west, align=left},
},
width=.65\linewidth, height=5.5cm,
]
\addplot[myorange,densely dashed] table [x = p, y expr = 1/\thisrow{d12m}]{\t};
\addlegendentry{Premium}
\addplot[myorange,densely dashed,forget plot] table [x = p, y expr = 1/\thisrow{d12h}]{\t};

\addplot[myblue,solid] table [x = p, y expr = 1/\thisrow{d6m}]{\t};
\addlegendentry{Regular}
\addplot[myblue,solid,forget plot] table [x = p, y expr = 1/\thisrow{d6h}]{\t};
\end{axis}
\end{tikzpicture}
  \caption{Impact of population distribution under two different loads, $\rho = 0.9$ (top plots) and $\rho = 0.99$ (bottom plots).}
  \label{fig:diff2}
\end{figure}

This figure confirms that the differentiation ratio decreases with the load,
and shows that the population distribution has a limited impact on performance.
When $ \rho = 0.9 $, both regular and premium jobs
suffer about 25\% rate degradation
between the best (premium only) and worst (regular only) scenarios.
When $ \rho = 0.99 $, the loss is limited to 14\% approximately.
This relative insensitivity of performance with respect to the population distribution is a positive result since
this distribution may not be known \emph{a priori} by the service provider.

Overall, these results show that randomized load balancing with variable degree of parallelism
is an efficient way of achieving service differentiation despite its simplicity.
We also saw in Section \ref{sec:randomizeds} that it guarantees stability as soon as the overall load is less than $1$.

For other degree parameters, numerical results (not displayed here) are qualitatively similar:  there is a differentiation proportional to the degree ratio when the load is low, which tends to fade at extreme loads.

\subsection{Impact of locality}
\label{subsec:locality}

We now study the impact of locality on the performance of randomized load balancing.
We consider a pool of $K$ homogeneous servers with unit service rate.
Each incoming job is assigned to a set of $d$ servers
chosen uniformly at random among the authorized assignments.
We consider the following assignment configurations,
which were studied in \S \ref{sec:homogeneous}, \S \ref{sec:ring-structure} and \S \ref{sec:randomized-load-balancing}:
\begin{description}
  \setlength\itemsep{0em}
  \item[Global] all sets of $d$ servers among $K$,
  \item[Ring] the sets of $d$ consecutive servers in a ring topology,
  \item[Line] the sets of $d$ consecutive servers in a line topology.
\end{description}
We first investigate the general performance hierarchy between these three configurations.

\paragraph{Costs of heterogeneity and locality}

As observed in \S \ref{sec:randomized-load-balancing},
the performance experienced by a job in a line scenario depends on its assignment.
Figure \ref{fig:load_study1} shows the mean service rate per class in the line,
compared to the overall mean service rate in each scenario.
Performance heterogeneity in the line increases with the load,
which leads to a degradation of the overall performance compared to the other scenarios.
We call this the \emph{cost of heterogeneity}.

\begin{figure}[h]
\centering
\subfloat[\label{fig:load5} Case $ \rho = 0.5 $]{
\pgfplotstableread[col sep=comma]{load_study_K_100_d_10_rho_0.5.txt}\cinq

\begin{tikzpicture}
\begin{axis}[
tikzLocality,
xlabel={Class},
ylabel={Mean service rate},
xmin=1, xmax=91, ymin = 5, ymax = 8,
legend columns=2,
legend style={
	row sep=-0.4cm,
	font = \footnotesize,
	inner ysep=-.1cm,
	cells={anchor=west, align=left},
	anchor=north,
	at={(axis description cs:0.5,.98)}
},
width=.45\linewidth, height=5cm,
]

\addplot[nonloc] table [x = c, y expr = 1/\thisrow{G}]{\cinq};
\addlegendentry{Global}
\addplot[line] table [x = c, y expr = 1/\thisrow{L}]{\cinq};
\addlegendentry{Line}
\addplot[ring] table [x = c, y expr = 1/\thisrow{R}]{\cinq};
\addlegendentry{Ring}
\addplot[line_c] table [x = c, y expr = 1/\thisrow{Li}]{\cinq};
\addlegendentry{Line (per class)}

\end{axis}
\end{tikzpicture}
 }
  \quad
  \subfloat[\label{fig:load9} Case $ \rho = 0.9 $]{
\pgfplotstableread[col sep=comma]{load_study_K_100_d_10_rho_0.9.txt}\neuf
\begin{tikzpicture}
\begin{axis}[
tikzLocality,
xlabel={Class},
xmin=1, xmax=91, ymax = 5,
legend style={
	row sep=-0.4cm,
	font = \footnotesize,
	cells={anchor=west, align=left},
},
width=.45\linewidth, height=5cm,
]
\addplot[nonloc] table [x = c, y expr = 1/\thisrow{G}]{\neuf};
\addplot[ring] table [x = c, y expr = 1/\thisrow{R}]{\neuf};
\addplot[line] table [x = c, y expr = 1/\thisrow{L}]{\neuf};
\addplot[line_c] table [x = c, y expr = 1/\thisrow{Li}]{\neuf};
\end{axis}
\end{tikzpicture}
  }
  \caption{Impact of locality ($ K = 100 $, $ d = 10 $)}
  \label{fig:load_study1}
\end{figure}
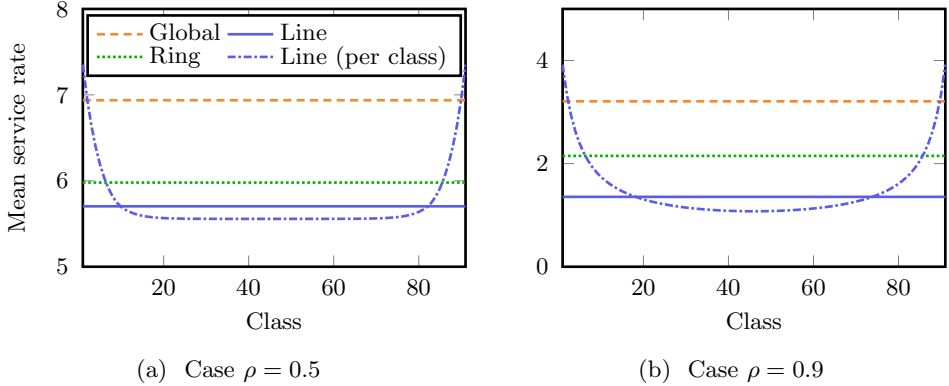

The ring scenario performs better than the line but not as well as the global case.
This is the \emph{cost of locality}, which we interpret as follows: the locality of assignments in line and ring scenarios reduces the diversity of classes compared to a global assignment;  it is more frequent to have two classes sharing a high number of servers,
which degrades the overall performance.

\paragraph{Impact of parameters}

To better understand these phenomena,
we let the parameters $K$, $d$ and $\rho$ vary around the following default values: $K = 100$, $d = 10$ and $\rho = 0.9$.
The results are shown in Figure \ref{fig:load_study2}.
First observe that the hierarchy between line, ring and global allocations is preserved throughout the experiments.

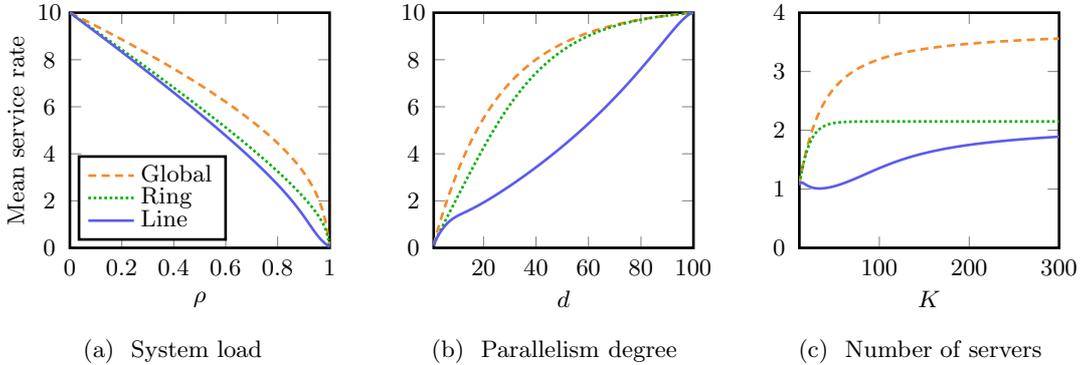
\begin{figure*}[h]
	\centering
	\subfloat[\label{fig:load_rho} System load]{
		\pgfplotstableread[col sep=comma]{load_study_rho_K_100_d_10.txt}\r
		\begin{tikzpicture}
		\begin{axis}[
		tikzLocality,
		xlabel={$ \rho $},
		ylabel={Mean service rate},
		xmin=0, xmax=1, ymax = 10,
		legend pos = south west,
		legend style={
			row sep=-0.4cm,
			font = \footnotesize,
			inner ysep=-.1cm,
			cells={anchor=west, align=left},
		},
		width=.34\linewidth, height=4.7cm,
		]		
		\addplot[nonloc] table [x = rho, y expr = 1/\thisrow{G}]{\r};
		\addlegendentry{Global}
		\addplot[ring] table [x = rho, y expr = 1/\thisrow{R}]{\r};
		\addlegendentry{Ring}
		\addplot[line] table [x = rho, y expr = 1/\thisrow{L}]{\r};
		\addlegendentry{Line}
		\end{axis}
		\end{tikzpicture}
	}
	\hfill
	\subfloat[\label{fig:load_d} Parallelism degree]{
		\pgfplotstableread[col sep=comma]{load_study_d_K_100_rho_0.9.txt}\d
		
		\begin{tikzpicture}
		\begin{axis}[
		tikzLocality,
		xlabel={$ d $},
		xmin=1, xmax=100, ymax = 10,
		width=.34\linewidth, height=4.7cm,
		]
		
		\addplot[nonloc] table [x = d, y expr = 1/\thisrow{G}]{\d};
		\addplot[ring] table [x = d, y expr = 1/\thisrow{R}]{\d};
		\addplot[line] table [x = d, y expr = 1/\thisrow{L}]{\d};
		
		\end{axis}
		\end{tikzpicture}
	}
	\hfill
	\subfloat[\label{fig:load_k} Number of servers]{
		\pgfplotstableread[col sep=comma]{load_study_K_Kmax_300_d_10_rho_0.9.txt}\k
		\begin{tikzpicture}
		\begin{axis}[
		tikzLocality,
		xlabel={$ K $},
		xmin=11, xmax=300, ymax = 4,
		width=.34\linewidth, height=4.7cm,
		]
		\addplot[nonloc] table [x = K, y expr = 1/\thisrow{G}]{\k};
		\addplot[ring] table [x = K, y expr = 1/\thisrow{R}]{\k};
		\addplot[line] table [x = K, y expr = 1/\thisrow{L}]{\k};
		
		\end{axis}
		\end{tikzpicture}
	}
	\caption{Overall impact of the parameters (default values: $ K = 100 $, $ d = 10 $, $ \rho = 0.9 $)}
	\label{fig:load_study2}
\end{figure*}

Figure \ref{fig:load_rho} shows the impact of the load $\rho$ on performance.
The mean service rate in the ring is close to that of the line when the load is low,
but it has the same asymptotic as the mean service rate in the global scenario when the load tends to $1$.
Intuitively, the cost of locality prevails at low load and impacts both the line and the ring;
when the load is higher, the cost of heterogeneity is the main source of performance degradation and impacts only the line.

Figure \ref{fig:load_d} studies the impact of the parallelism degree $d$ on performance.
First observe that the mean service rate increases with the degree in each case.
This increase is much faster in the global and ring scenarios than in the line scenario.
Our interpretation is the following.
In the line scenario,
the total number $K-d+1$ of classes decreases with $d$, hence performance suffers from a lack of diversity in the assignment
compared to the global and ring cases.

Lastly, Figure \ref{fig:load_k} gives the evolution of the performance as a function of the number $K$ of servers.
It was proved in \cite{G17-1} that the mean service rate in the global scenario has a limit when $K$ tends to infinity.
This is consistent with the results of Figure \ref{fig:load_k},
which suggests that a limit also exists in the ring and line scenarios.
Note that the convergence is quite fast in the ring, and non-monotonic in the line.
The behavior for the line can be intuitively explained by the heterogeneity of the number of classes per server: for example, when $ K $ is close to $ d $, a majority of the servers can serve all $ K-d+1 $ classes; when $ K=2d $, there are exactly two servers that can serve $ k $ classes for each $ k = 1,\ldots,d $, so there is a lot of heterogeneity between servers. For larger values of $ K $, the cost of heterogeneity in the line fades away as it becomes a border effect from classes and servers located near the edges; this explains why it seems that the line and ring scenarios share the same limit: as $ K $ increases, only the cost of locality prevails.

All these results show that local load balancing has a cost in terms of performance that depends on the parameters.
However, keeping in mind that a local allocation may be more simple to implement in a real system and has no impact on the stability,
we believe that it can be a viable option.
Remark that whenever possible, a ring structure should be preferred to a line structure.

\section{Conclusion}\label{sec:conclusion}

In this paper, we have considered a resource pool model
where operational constraints like data availability, locality, or allowed degree of parallelism 
are represented by an assignment graph between job classes and servers.
Resources are allocated by applying balanced fairness under these constraints.
Although ideal, this resource allocation can be implemented in practice
by some sequential \gls{FCFS} scheduling at each server.
Our main contribution is a new recursive formula to compute the performance metrics under an arbitrary assignment graph.
The key ingredient is the observation that the idling probability of each server can be derived
by comparing the behavior of the system with and without this server.

Although the complexity of our formula is exponential in the number of servers in general,
it provides a unified framework for analyzing balanced fairness in resource pools,
which allows to simplify the study in many practically interesting cases.
Specifically, we have identified two  classes of models where the equivalence of the servers or the structure of the assignment graph
lead to simplifications making the complexity polynomial, enabling an exact evaluation of their behavior.

For future works,
we would like to identify other classes of resource pools where performance is made tractable by our formula.
We are also interested in deriving more intuition on the impact of the assignment graph on  performance.
We hope that our work will stimulate further studies on resource pools under balanced fairness.

\bibliographystyle{abbrv}

\end{document}